\documentclass[a4paper,11pt]{article}
\usepackage[utf8]{inputenc}
\usepackage{fullpage}
\usepackage{amsmath}
\usepackage{bussproofs}
\usepackage{proof}
\usepackage{enumerate}
\usepackage{color}
\usepackage{lineno}
\usepackage{amsthm}
\usepackage{amsfonts}
\usepackage{todonotes}
\usepackage{tikz}
\tikzset{node distance=2cm, auto}
\usepackage{thmtools,thm-restate}

\usepackage{graphicx}
\usepackage{relsize}

\usepackage{url}
\usepackage{enumitem}
\usepackage{xcolor}
\usepackage{mathrsfs}
\usepackage{colortbl}
\usepackage[all]{xy}

\usepackage{soul}

\theoremstyle{plain} 
\newtheorem{theorem}{Theorem}
\newtheorem{proposition}[theorem]{Proposition}
\newtheorem{observation}[theorem]{Observation}

\newtheorem{lemma}[theorem]{Lemma}
\newtheorem{corollary}[theorem]{Corollary}
\newtheorem{question}[theorem]{Question}
\theoremstyle{definition}
\newtheorem{definition}[theorem]{Definition}
\newtheorem{example}[theorem]{Example}
\newtheorem{remark}[theorem]{Remark}

\def\phi{\varphi}

\newcommand{\calA}{\mathcal{A}}
\newcommand{\calB}{\mathcal{B}}
\newcommand{\calC}{\mathcal{C}}
\newcommand{\calD}{\mathcal{D}}
\newcommand{\calI}{\mathcal{I}}
\newcommand{\calM}{\mathcal{M}}
\newcommand{\calS}{\mathcal{S}}

\newcommand{\myvec}[1]{\overline{#1}}

\newcommand{\Land}{\bigwedge}

\newcommand{\impl}{\rightarrow}
\newcommand{\liff}{\leftrightarrow}
\newcommand{\CNF}{\mathrm{CNF}}
\newcommand{\mCNF}{\mathrm{mCNF}}
\newcommand{\seq}{\Rightarrow}
\newcommand{\Mseq}[5][{}]{{#2} ; {#3} \overset{#1}{\seq} {#4} ; {#5}} 
\newcommand{\unsubst}[2]{[{#1}\backslash{#2}]} 
\newcommand{\LK}{\ensuremath{\mathbf{LK}}}
\newcommand{\LKminus}{\ensuremath{\LK^-}}
\newcommand{\LKat}{\ensuremath{\LK^\mathrm{at}}}
\newcommand{\LKlit}{\ensuremath{\LK^\mathrm{lit}}}
\newcommand{\LKm}{\ensuremath{\LK^\mathrm{m}}} 
\newcommand{\subs}{\leq_\mathrm{ss}} 
\newcommand{\weight}{\mathrm{w}}
\newcommand{\degree}{\mathrm{d}}

\usepackage{fixme}
\fxsetup{
    status=draft,
    author=,
    layout=footnote, 
    theme=color
}
\definecolor{fxnote}{rgb}{0.8000,0.0000,0.0000}

\title{On the Completeness of Interpolation Algorithms}
\author{Stefan Hetzl \and Raheleh Jalali\footnote{Supported by Czech Science Foundation Grant No. 22-06414L.}}

\begin{document}

\maketitle
\begin{abstract}
Craig interpolation is a fundamental property of classical and non-classic logics
with a plethora of applications from philosophical logic to computer-aided verification. The question of which interpolants can be obtained from an interpolation algorithm is of profound importance. Motivated by this question, we initiate the study of completeness properties of interpolation algorithms. An interpolation algorithm $\calI$ is \emph{complete} if, for every semantically possible interpolant $C$ of an implication $A \impl B$, there is a proof $P$ of $A \impl B$ such that $C$ is logically equivalent to $\calI(P)$. We establish incompleteness and different kinds of completeness results for several standard algorithms for resolution and the sequent calculus for propositional, modal, and first-order logic.
\end{abstract}




\textbf{keywords:} (Craig) interpolation, proof theory, sequent calculus, resolution

\maketitle

\section{Introduction}


Interpolation is one of the most fundamental properties of classical and
 many non-classical logics.
Proved first by Craig in the 1950ies~\cite{Craig57Three}, the interpolation
 theorem has remained central to logic ever since and continues to find new 
 applications, see, e.g., \cite{Mancosu08Introduction} to  catch a glimpse
 of the remarkable breadth of this seminal result.
It has widespread applications throughout all areas of logic in philosophy, 
 mathematics, and computer science.
For example, the interpolation theorem has strong connections to the
 foundations of mathematics through its use for proving Beth's definability
 theorem~\cite{Beth53Padoa}.
It underlies the technique of feasible interpolation for proving lower bounds
 in proof complexity~\cite{Krajicek94Lower}.
It has also become a versatile and indispensable tool in
 theory reasoning~\cite{nelson,tinelli}, model checking~\cite{McMillan18Interpolation}, and knowledge representation
 \cite{amirpartition}.

A question of profound importance, from a theoretical as well as a practical
 point of view, is that of the expressive power of interpolation algorithms:
Given a particular interpolation algorithm, which interpolants can it compute?
This question is relevant for all applications of interpolation and it is
 gaining importance with the increasing number of applications.
In particular, it has been in the focus of attention in the CAV community
 where a lot of work, see, e.g.~\cite{Jhala05Interpolant,Jhala07Interpolant,
 Jhala06Practical,DSilva10Interpolant,Weissenbacher12Interpolant,
 Kovacs09Interpolation,Hoder12Playing}, has been devoted to 
 developing algorithms that compute ``good" interpolants, i.e., such interpolants that have properties, e.g., in terms of logical strength, that make them
 well suited for the application at hand.
However, from a theoretical point of view, this question has not been
 investigated to a satisfactory degree yet.

In this paper, we set out to gauge the expressive power of different interpolation
 algorithms by carrying out a thorough theoretical analysis.
To this aim we initiate the study of completeness properties of interpolation
 algorithms.
An interpolation algorithm outputs an interpolant of an implication $A\impl B$
 when given a proof of $A\impl B$ (in a suitable proof system).
We say that an interpolation algorithm $\calI$ is \emph{complete} if, for every 
 semantically possible interpolant $C$ of an implication $A \impl B$, there is a 
 proof $P$ of $A \impl B$ such that $\calI(P)$ is logically equivalent to $C$.
%
%
The practical relevance of a completeness result is that it provides
 a guarantee that, at least in principle, the algorithm allows to find
 the ``good" interpolants, whatever that may mean in the concrete application under
 consideration.


In this paper we study the completeness of the standard interpolations algorithms for the
 sequent calculus and resolution, two of the most important proof
 systems in mathematical logic and computer science.
We establish the following results: after introducing the
 necessary notions and definitions in  Section~\ref{sec.prelim} we will show
 in Section~\ref{sec.simple_incomp} that the standard interpolation algorithms for resolution and cut-free sequent calculus for classical propositional logic are incomplete.
On the other hand, if we weaken the completeness property, it is possible to
 obtain a positive result for the cut-free sequent calculus:
in Section~\ref{sec.comp_prun_subs} we will show that the standard interpolation
 algorithm for cut-free sequent calculus is complete up to subsumption for
 pruned interpolants.
This result requires a subtle argument that traces the
 development of interpolants during cut-elimination.
Furthermore, in Section~\ref{sec.LKat} we show that in the sequent calculus 
 with cut, already with atomic cut, the standard interpolation algorithm is 
 complete.
We then leave the realm of propositional logic to extend our results
 to normal modal logics in Section~\ref{sec.modal} and to classical
 first-order logic in Section~\ref{sec.fol}.
While the results for propositional logic carry over to first-order logic,
 we also find a new source of incompleteness on the first-order level which
 implies that standard algorithms for first-order interpolation are
 incomplete, even in the sequent calculus with atomic cuts.
Last, but not least, in Section~\ref{sec.Beth} we show that completness properties
 of interpolation algorithms translate directly to completeness properties
 of Beth's definability theorem.

\section{Preliminaries}\label{sec.prelim}

\subsection{Formulas}

For the preliminaries, the reader may consult \cite{Buss, negri, Troelstra}.
Let the propositional language $\mathcal{L}_p= \{\bot, \wedge, \vee,  \neg\}$. \emph{Propositional variables} or \emph{atoms} are denoted by $p,q, \dots$, possibly with subscripts. \emph{Formulas} are denoted by $A, B, \dots$ and defined as usual. Define $A \to B := \neg A \vee B$ and $\top:= \bot \to \bot$ for any formulas $A$ and $B$ in the language $\mathcal{L}_p$. A \emph{literal} $\ell$ is either an atom or a negation of an atom. We consider $\bot$ as an atom and $\top$ as a literal. If $\ell=p$, by $\bar{\ell}$ we mean $\neg p$ and if $\ell=\neg p$, by $\bar{\ell}$ we mean $p$. 
A \emph{clause} $C$ is a finite disjunction of literals $C=\ell_1 \vee \dots \vee \ell_n$, sometimes simply written as $C=\{\ell_1, \dots , \ell_n\}$. The clause $C \cup \{\ell\}$ is often abbreviated by $C, \ell$. We assume two clauses are equal if they have the same set of literals.
A \emph{positive} literal is an unnegated atom and a \emph{negative} literal is a negated atom. 
By convention, $\wedge \emptyset := \top$ and $\vee \emptyset := \bot$.
For a formula $A$, the set of its \emph{variables}, denoted by $V(A)$, is defined recursively; $V(\bot)=\emptyset$; $V(p)=p$ for any atom $p$; $V(\neg A)=V(A)$, and $V(A \circ B) = V(A) \cup V(B)$ for $\circ \in \{\wedge, \vee\}$. For a multiset $\Gamma$ define $V(\Gamma) = \bigcup_{\gamma \in \Gamma} V(\gamma)$.

\begin{definition}\label{Dfn: CIP logic}
A logic $L$ has the \emph{Craig Interpolation Property} (\emph{CIP}) if for any formulas $A$ and $B$ if $A \to B \in L$ then there exists a formula $C$ such that $V(C) \subseteq V(A) \cap V(B)$ and $A \to C \in L$ and $C \to B \in L$. 
\end{definition}

\subsection{Resolution}

Propositional \emph{resolution}, $R$, is one of the weakest proof systems. However, it is very important in artificial intelligence as it is useful in automated theorem proving and SAT solving. Resolution operates on clauses. 

By a \emph{clause set} we mean a set $\mathcal{C} =\{C_1, \dots , C_n\}$ of clauses $C_i =\{\ell_{i1}, \dots, \ell_{ik_i}\}$ and the \emph{formula interpretation} of $\calC$ is $\bigwedge_{i=1}^n \bigvee_{j=1}^{k_i} \ell_{ij}$.
We say a clause set is \emph{logically equivalent} to a formula $\phi$ when its formula interpretation is logically equivalent to $\phi$.
A formula is in \emph{conjunctive normal form} if it is a conjunction of disjunction of literals, i.e., has the form $\bigwedge_{i=1}^{n} \bigvee_{j=1}^{k_i} \ell_{ij}$ for some $1 \leq n$ and $1 \leq k_i$ for each $1 \leq i \leq n$, where $\ell_{ij}$'s are literals. Equivalently, we may use the clause set $\calC=\{C_i \mid 1\leq i \leq n \}$, where each $C_i= \ell_{i1} \vee \dots \vee \ell_{ik_i}$. We use these two formats interchangeably for the conjunctive normal form of a formula.

A \emph{resolution proof}, also called a \emph{resolution refutation}, shows the unsatisfiability of a set of initial clauses by starting
with these clauses and deriving new clauses by the \emph{resolution rule}
\begin{prooftree}
 \AxiomC{$C \cup \{p\}$} 
  \AxiomC{$D \cup \{\neg p\}$} 
  \BinaryInfC{$ C \cup D$}
\end{prooftree}
until the empty clause is derived, where $C$ and $D$ are clauses. This rule is obviously sound, i.e., if a truth assignment satisfies both the premises, then the conclusion is also satisfied by the same assignment. Given a set of clauses, it is possible to use the resolution rule and derive new clauses. Specifically, we may derive the empty clause, denoted by $\bot$. Thus, we can interpret resolution as a refutation system; instead of proving a formula $A$ is true we prove that $\neg A$ is unsatisfiable. We write $\neg A$ as a conjunction of disjunction of literals and take each conjunct as a clause. Then, from these clauses, we derive the empty clause using the resolution rule only. It is also possible to add the \emph{weakening rule} to the resolution system:
\begin{prooftree}
\AxiomC{$C$}
\UnaryInfC{$C \cup D$}
\end{prooftree}
for arbitrary clauses $C$ and $D$. The new system is called \emph{resolution with weakening}.

\noindent \textbf{Interpolation algorithm for resolution \cite{krajicek,pudlak}:} Suppose a resolution proof $P$ of the empty clause from the clauses $A_i(\bar{p}, \bar{q})$ and $B_j(\bar{p}, \bar{r})$ is given, where $i \in I$, $j \in J$, and $\bar{p}, \bar{q}, \bar{r}$ are disjoint sets of atoms. The only atoms common between the two sets of clauses are $\bar{p}$. Define a ternary connective $sel$ as $sel(\bot,x,y)=x$ and $sel(\top, x, y)=y$. For instance, by definition, we have $sel(A, \bot, \top)=A$, $sel(A, \top, \bot)= \neg A$, and $sel (A, x, y)= (\neg A \to x) \wedge (A \to y)= (A \vee x) \wedge (\neg A \vee y)$, for formulas $A$, $x$, and $y$.
The interpolation algorithm operates as follows. Assign the constant $\bot$ to clauses $A_i$ for each $i \in I$ and assign $\top$ to clauses $B_j$ for $j \in J$. Then:
 
 (1) Suppose the resolution rule is of the form
\begin{prooftree}
    \AxiomC{$\Gamma, p_k$}
    \AxiomC{$\Delta, \neg p_k$}
    \BinaryInfC{$\Gamma, \Delta$}
\end{prooftree}
where $p_k \in \bar{p}$. If we have assigned $x$ to the premise $\Gamma, p_k$ and $y$ to $\Delta, \neg p_k$ then we assign $z= sel(p_k,x,y)$ to the conclusion $\Gamma, \Delta$.

 (2) Suppose the resolution rule is of the form
\begin{prooftree}
    \AxiomC{$\Gamma, q_k$}
    \AxiomC{$\Delta, \neg q_k$}
    \BinaryInfC{$\Gamma, \Delta$}
\end{prooftree}
where $q_k \in \bar{q}$. If we have assigned $x$ to the premise $\Gamma, q_k$ and $y$ to $\Delta, \neg q_k$, then we will assign $x \vee y$ to the conclusion $\Gamma, \Delta$.

 (3)  Suppose the resolution rule is of the form
\begin{prooftree}
    \AxiomC{$\Gamma, r_k$}
    \AxiomC{$\Delta, \neg r_k$}
    \BinaryInfC{$\Gamma, \Delta$}
\end{prooftree}
where $r_k \in \bar{r}$. If we have assigned $x$ to the premise $\Gamma, r_k$ and $y$ to $\Delta, \neg r_k$, then we will assign $x \wedge y$ to the conclusion $\Gamma, \Delta$.

This algorithm outputs an interpolant of the valid formula $A \to B$, given a resolution refutation of the unsatisfiable formula $A \wedge \neg B$. 

The interpolation algorithm for resolution with weakening is defined as above except that when the weakening rule is used
\begin{prooftree}
\AxiomC{$\Gamma$}
\UnaryInfC{$\Gamma, \Delta$}
\end{prooftree}
and we have assigned $x$ to the premise $\Gamma$, we will also assign $x$ to the conclusion $\Gamma, \Delta$.

\begin{theorem}\cite{krajicek,pudlak}
Let $\pi$ be a resolution refutation of the set of clauses $\{A_i(\bar{p}, \bar{q}) \mid i \in I\}$ and $\{B_j(\bar{p}, \bar{q}) \mid j \in J\}$. Then, the interpolation algorithm outputs an interpolant for the valid formula $\bigwedge_{i \in I} A_i(\bar{p}, \bar{q}) \to \bigvee_{j \in J} \neg B_j(\bar{p}, \bar{q})$.
\end{theorem}
In refutational theorem provers, interpolation is also often formulated in terms
 of reverse interpolants.
A {\em reverse interpolant} of an unsatisfiable formula $A\land B$ is a formula
 $C$ with $V(C) \subseteq V(A)\cap V(B)$ such that $A \models C$ and $B\models \neg C$.
Note that $C$ is a reverse interpolant of $A\land B$ iff it is an interpolant
 of $A\impl \neg B$.

\begin{table}[t]
 \centering
\begin{tabular}{c c}
\vspace{5pt} 

$A \Rightarrow A$ \; \small{(Ax)} & $\bot \Rightarrow $\; \small{($\bot$)}\\

\vspace{5pt}

\AxiomC{$\Gamma \Rightarrow \Delta$}
 \RightLabel{\small{$(L w)$} }
\UnaryInfC{$A, \Gamma \Rightarrow \Delta$}
\DisplayProof
&  
\AxiomC{$\Gamma \Rightarrow \Delta$}
 \RightLabel{\small{$(R w)$} }
\UnaryInfC{$\Gamma \Rightarrow \Delta, A$}
 \DisplayProof\\

\vspace{5pt}

\AxiomC{$A, A, \Gamma \Rightarrow \Delta$}
 \RightLabel{\small{$(L c)$} }
\UnaryInfC{$A, \Gamma \Rightarrow \Delta$}
\DisplayProof
&  
\AxiomC{$\Gamma \Rightarrow \Delta, A, A$}
 \RightLabel{\small{$(R c)$} }
\UnaryInfC{$\Gamma \Rightarrow \Delta, A$}
 \DisplayProof\\
 
\vspace{5pt}

\AxiomC{$A,\Gamma \Rightarrow \Delta$}
 \RightLabel{\small{$(L \wedge_1)$} }
\UnaryInfC{$A \wedge B, \Gamma \Rightarrow \Delta$}
\DisplayProof
&
\AxiomC{$B,\Gamma \Rightarrow \Delta$}
 \RightLabel{\small{$(L \wedge_2)$} }
\UnaryInfC{$A \wedge B, \Gamma \Rightarrow \Delta$}
\DisplayProof\\
 
\vspace{5pt}
\AxiomC{$\Gamma \Rightarrow \Delta, A$}
 \AxiomC{$\Gamma \Rightarrow \Delta, B$}
 \RightLabel{\small{$(R \wedge)$} }
\BinaryInfC{$\Gamma \Rightarrow \Delta, A \wedge B$}
 \DisplayProof
 &
\AxiomC{$\Gamma \Rightarrow \Delta, A$}
 \RightLabel{\small{$(R \vee_1)$} }
 \UnaryInfC{$\Gamma \Rightarrow \Delta, A \vee B$}
 \DisplayProof\\

\vspace{5pt}

\AxiomC{$\Gamma \Rightarrow \Delta,  B$}
 \RightLabel{\small{$(R \vee_2)$} }
 \UnaryInfC{$\Gamma \Rightarrow \Delta, A \vee B$}
 \DisplayProof
 &
 \AxiomC{$A, \Gamma \Rightarrow \Delta$}
 \AxiomC{$B, \Gamma \Rightarrow \Delta$}
 \RightLabel{\small{$(L \vee)$} }
\BinaryInfC{$A \vee B, \Gamma \Rightarrow \Delta$}
 \DisplayProof \\


 \vspace{5pt}

\AxiomC{$\Gamma \Rightarrow \Delta, A$}
 \RightLabel{\small{$(L \neg)$} }
\UnaryInfC{$\neg A , \Gamma \Rightarrow \Delta$}
\DisplayProof &
\AxiomC{$\Gamma, A \Rightarrow \Delta$}
 \RightLabel{\small{$(R \neg)$} }
\UnaryInfC{$ \Gamma \Rightarrow \Delta, \neg A$}
\DisplayProof \\
\AxiomC{$\Gamma \Rightarrow \Delta, A$}
 \AxiomC{$A, \Gamma \Rightarrow \Delta$}
 \RightLabel{\small{$(cut)$} }
\BinaryInfC{$\Gamma \Rightarrow \Delta $} 
\DisplayProof 

    \end{tabular}
    \caption{Propositional $\LK$. 
    The formula $A$ in the cut rule is called the \emph{cut formula}.} 
    \label{table: LK}
\end{table}

\subsection{Sequent calculus}\label{ssec.seq_calc}

A \emph{sequent} is an expression of the form $\Gamma \Rightarrow \Delta$ where $\Gamma$, called the \emph{antecedent}, and $\Delta$, called the \emph{succedent}, are multisets of formulas and the \emph{formula interpretation} of the sequent is $\bigwedge \Gamma \to \bigvee \Delta$. 
In the propositional sequent calculus $\LK$ each proof is represented as a tree. 
The nodes correspond to sequents and the root of the proof tree, at the bottom, is the \emph{end-sequent} and it is the sequent proved by the proof. The topmost nodes of the tree, the leaves, are called the \emph{initial sequents} or \emph{axioms}. Apart from the axioms, each sequent in an $\LK$ proof must be inferred from one of the inference rules provided in Table \ref{table: LK}. An \emph{inference rule} is an expression of the form $\frac{S_1}{S}$ or $\frac{S_1 \; S_2}{S}$ indicating that the conclusion $S$ is inferred from the premise $S_1$ or the premises $S_1$ and $S_2$. In Table \ref{table: LK}, rules are schematic, $A$ and $B$ denote arbitrary formulas, and $\Gamma$ and $\Delta$ denote arbitrary multisets of formulas. The first row in Table \ref{table: LK} are axioms. The weakening and contraction rules and the cut rule are \emph{structural}   inference rules and the rest \emph{logical}. 
By \emph{length} of a formula $A$, denoted by $|A|$, we mean the number of symbols in it. Similarly, we define the length of a multiset, sequent, and a proof. The \emph{size} of a tree is the number of nodes in it. The \emph{depth} or \emph{height} of a proof tree is the maximum length of the branches in the tree, where the \emph{length} of a branch is the number of nodes in the branch minus one.

We will consider the $G1$ systems for our logics, as presented
 in~\cite[Section 3.1]{Troelstra}. 
However, the results of this paper do not depend significantly on the concrete
 version of the sequent calculus.
The system presented in Table \ref{table: LK} is the propositional $\LK$. In each rule, the upper sequent(s) is (are) \emph{premise(s)}, and the lower sequent is the \emph{conclusion}.
$\Gamma$ and $\Delta$ are called the \emph{context}. The formula in the conclusion of the rule not belonging to the contexts is called the \emph{main} formula. The formulas in the premises not in the contexts are called the \emph{active} or \emph{auxiliary} formulas. 
The cut rule is called \emph{atomic} when the cut formula is an atom or $\bot$ or $\top$. If we take $\Gamma=\Gamma_1 \cup \Gamma_2$ and $\Delta_1 \cup \Delta_2$, then the cut rule is called \emph{monochromatic} if either $V(A) \subseteq V(\Gamma_1) \cup V(\Delta_1)$ or $V(A) \subseteq V(\Gamma_2) \cup V(\Delta_2)$. A cut rule is \emph{analytic} if the cut formula $A$ is a subformula of a formula occurring in the conclusion. Clearly, any analytic cut is also monochromatic.
Denote $\LK$ with only atomic (resp. monochromatic) cuts by $\LKat$ (resp. $\LKm$), denote $\LK$ with cuts only on literals by $\LKlit$, and denote cut-free $\LK$ by $\LKminus$. 



\begin{remark}\label{Rem: monochromatic atomic cut}
Without loss of generality, we can consider each atomic cut in $\LKat$ or each cut on a literal in $\LKlit$ as monochromatic. Let us explain the case for $\LKat$, the case for $\LKlit$ is similar. If $\pi$ is a proof of $\Gamma \Rightarrow \Delta$ in $\LKat$, then there exists a proof of $\Gamma \Rightarrow \Delta$ in $\LKat$ where each cut in the proof is atomic and monochromatic. The sketch of the proof is as follows: We use induction on the number of non-monochromatic cuts in the proof $\pi$. Take the topmost non-monochromatic instance of the cut rule used in the proof. Since it is non-monochromatic, it means that the atomic cut formula does not appear in the conclusion of the cut rule. Choose an atom $q$ in the conclusion. Everywhere in this subproof replace the atomic cut formula by $q$. Continue until all the cuts become monochromatic. By this observation, from now on we assume any atom appearing in a proof $\pi$ of $\Gamma \Rightarrow \Delta$ in $\LKat$ is in $V(\Gamma) \cup V(\Delta)$.
\end{remark}
\noindent \textbf{Maehara interpolation algorithm for  $\LKat$, $\LKlit$ and $\LKm$:} Denote the Maehara interpolation algorithm by $\calM$. A \emph{split sequent} is an expression of the from $\Gamma_1;\Gamma_2 \Rightarrow \Delta_1;\Delta_2$ such that $\Gamma_1,\Gamma_2 \Rightarrow \Delta_1,\Delta_2$ is a sequent. The partition $\Gamma_1, \Gamma_2=\Gamma$ and $\Delta_1, \Delta_2=\Delta$ is called the \emph{Maehara partition} of the sequent $\Gamma \Rightarrow \Delta$. A formula $A$ is on the \emph{left-hand side} (resp. \emph{right-hand side}) of the Maehara partition if $A \in \Gamma_1 \cup \Delta_1$ (resp. $A \in \Gamma_2 \cup \Delta_2$). Let $G \in \{\LKat, \LKlit, \LKm\}$. We define the Maehara algorithm for $G$.
Suppose a valid sequent $\Gamma \Rightarrow \Delta$ is given. Let $\pi$ be a proof of the split sequent $\Gamma_1;\Gamma_2 \Rightarrow \Delta_1;\Delta_2$ in  $G$. We define the formula $\calM(\pi)=C$, called the \emph{interpolant}, recursively. Let us denote the interpolant $C$ for $\Gamma_1; \Gamma_2 \Rightarrow \Delta_1; \Delta_2$ by $\Gamma_1; \Gamma_2 \overset{C}{\Longrightarrow} \Delta_1; \Delta_2$.

$\bullet$ If $\pi$ is an axiom, it is either (Ax) or ($\bot$). Then, $\calM(\pi)$ is defined as follows based on the occurrence of the atom $p$ or $\bot$ in the partitions:
\begin{center}
\begin{tabular}{c c c}
 $p;  \overset{\bot}{\Longrightarrow}  p; $ & $; p \overset{\top}{\Longrightarrow} ; p$     &  $ \bot; \overset{\bot}{\Longrightarrow} ; $\\
 $ p;  \overset{p}{\Longrightarrow} ;  p$    & $;  p \overset{\neg p}{\Longrightarrow}  p; $ & $;  \bot \overset{\top}{\Longrightarrow} ; $
\end{tabular}
\end{center}

$\bullet$ If the last rule $\mathcal{R}$ used in $\pi$ is one of the one-premise rules, i.e., $(R \vee_1)$, $(R \vee_2)$, $(L \wedge_1)$, $(L \wedge_2)$, $(Lw)$, $(Rw)$, $(L c)$, $(R c)$, $(L \neg)$, or $(R \neg)$, then the interpolant of the premise also works as the interpolant for the conclusion, i.e., define $\calM(\pi)=\calM(\pi')$ where $\pi'$ is the proof of the premise of $\mathcal{R}$. 

$\bullet$ If the last rule in $\pi$ is $(R \wedge)$, there are two cases based on the occurrence of the main formula in the conclusion:
\begin{prooftree}
\AxiomC{$\Gamma_1; \Gamma_2 \overset{C}{\Longrightarrow} \Delta_1, A ; \Delta_2$}
 \AxiomC{$\Gamma_1; \Gamma_2 \overset{D}{\Longrightarrow} \Delta_1, B ; \Delta_2$}
 \BinaryInfC{$\Gamma_1; \Gamma_2 \overset{C \vee D}{\Longrightarrow} \Delta_1, A \wedge B ; \Delta_2$}    
\end{prooftree}
or 
\begin{prooftree}
 \AxiomC{$\Gamma_1; \Gamma_2 \overset{C}{\Longrightarrow} \Delta_1 ; A, \Delta_2$}
 \AxiomC{$\Gamma_1 ; \Gamma_2 \overset{D}{\Longrightarrow} \Delta_1 ; B, \Delta_2$}
 \BinaryInfC{$\Gamma_1; \Gamma_2 \overset{C \wedge D}{\Longrightarrow} \Delta_1 ; A \wedge B, \Delta_2$}    
\end{prooftree}

$\bullet$ If the last rule in $\pi$ is $(L \vee)$, then:
\begin{prooftree}
\AxiomC{$\Gamma_1, A;  \Gamma_2 \overset{C}{\Longrightarrow} \Delta_1 ; \Delta_2$}
 \AxiomC{$\Gamma_1, B; \Gamma_2 \overset{D}{\Longrightarrow} \Delta_1; \Delta_2$}
 \BinaryInfC{$\Gamma_1, A \vee B; \Gamma_2 \overset{C \vee D}{\Longrightarrow} \Delta_1; \Delta_2$}    
\end{prooftree}
or
\begin{prooftree}
  \AxiomC{$\Gamma_1; A, \Gamma_2 \overset{C}{\Longrightarrow} \Delta_1 ; \Delta_2$}
 \AxiomC{$\Gamma_1; B, \Gamma_2 \overset{D}{\Longrightarrow} \Delta_1; \Delta_2$}
 \BinaryInfC{$\Gamma_1; A \vee B,  \Gamma_2 \overset{C \wedge D}{\Longrightarrow} \Delta_1; \Delta_2$}   
\end{prooftree}
 
$\bullet$ Suppose the last rule in $\pi$ is an instance of a cut rule in $G$. Let $A$ be the cut formula. Denote $V_1= V(\Gamma_1 \cup \Delta_1)$ and $V_2= V(\Gamma_2 \cup \Delta_2)$. Then, either $A \in V_1$ or $A \in V_2$ or $A \in V_1 \cap V_2$.
If $A \in V_1$, define
\begin{prooftree}
 \AxiomC{$\Gamma_1; \Gamma_2 \overset{C}{\Longrightarrow} \Delta_1, A ; \Delta_2$}
 \AxiomC{$\Gamma_1, A; \Gamma_2 \overset{D}{\Longrightarrow} \Delta_1 ; \Delta_2$}
 \BinaryInfC{$\Gamma_1; \Gamma_2 \overset{C \vee D}{\Longrightarrow} \Delta_1 ; \Delta_2$}    
\end{prooftree}
If $A \in V_2$, define
\begin{prooftree}
 \AxiomC{$\Gamma_1; \Gamma_2 \overset{E}{\Longrightarrow} \Delta_1 ; A, \Delta_2$}
 \AxiomC{$\Gamma_1 ; \Gamma_2, A \overset{F}{\Longrightarrow} \Delta_1 ; \Delta_2$}
 \BinaryInfC{$\Gamma_1; \Gamma_2 \overset{E \wedge F}{\Longrightarrow} \Delta_1 ; \Delta_2$}    
\end{prooftree} 
And if $A$ is in both, choose either case.
 

\begin{restatable}{theorem}{ThmSoundnessMaehara}
\label{thm: soundness of Maehara}
Let $\pi$ be a proof of $A \to B$ in $G \in \{\LKat, \LKlit, \LKm\}$. Then, $\calM(\pi)$ outputs an interpolant of the formula $A \to B$ with $|\calM(\pi)| \leq |\pi|$.
\end{restatable}
\begin{proof}[Proof in the appendix.]
\end{proof}

A formula is in \emph{negation
normal form (NNF)} when the negation is only allowed on atoms and the other connectives in the formula are disjunctions and conjunctions. By an easy investigation of the form of the interpolants constructed in the Maehara algorithm, we see that in binary rules, based on the position of the active formulas in the premises, the interpolant of the conclusion will be either the conjunction or disjunction of the interpolants of the premises.
\begin{corollary}\label{Cor: NNF}
The interpolants constructed via the Maehara algorithm are in NNF. 
\end{corollary}
\begin{remark}
By carefully investigating the definition of the Maehara algorithm for the sequent calculus $G \in \{\LKat, \LKlit, \LKm\}$, we see that the following happens. If the cut formula is on the left-hand side (resp. right-hand side) of the Maehara partition, then the interpolant of the conclusion of the cut rule is the disjunction (resp. conjunction) of the interpolants of the premises of the cut rule. We will use this observation in the following results.
\end{remark}

\section{Simple incompleteness results}\label{sec.simple_incomp}

We start with a simple observation.
Some interpolation algorithms return only formulas of a particular shape.
For example, as observed in Corollary~\ref{Cor: NNF},
 Maehara algorithm $\calM$ only returns formulas in NNF.
This immediately yields a first incompleteness result as follows.
\begin{definition}
An interpolation algorithm $\calI$ is called {\em syntactically complete}
 if for every valid implication $A\impl B$ and every interpolant $C$
 of $A\impl B$ there is a proof $\pi$ such that $C = \calI(\pi)$.
\end{definition}
\begin{observation}
$\calM$ is syntactically incomplete.
\end{observation}
\begin{proof}
$\neg \neg p$ is an interpolant of $p \impl p$ but, since it is not in
 NNF, there is no $\pi$ such that $\calM(\pi) = \neg \neg p$.
\end{proof}
This simple observation makes clear that in many cases we will only
 obtain an interesting question if we ask for completeness up to
 some equivalence relation coarser than syntactic equality.
The first candidate that comes to mind, and the most important case
 we will deal with in this paper, is logical equivalence.
\begin{definition}\label{def.ipol_algo_complete}
An interpolation algorithm $\calI$ is called {\em (semantically) complete}
 if for every valid implication $A\impl B$ and every interpolant $C$
 of $A\impl B$ there is a proof $\pi$ s.t.\ $C$ is logically equivalent
 to $\calI(\pi)$.
\end{definition}
Since, in this paper, we will mostly deal with semantic completeness
 we will usually simply say {\em completeness} instead.

A useful basis for negative results is the implication
 $p\land q \impl p\lor q$.
It has the four interpolants $p\land q, p, q, p\lor q$ but, in
 proof systems that do not allow overly redundant proofs there are
 essentially only two different proofs: one that proceeds via $p$
 (and has $p$ as interpolant) and one that proceeds via $q$ (and has
 $q$ as interpolant).
This can be used to obtain incompleteness results for interpolation
 in the cut-free sequent calculus and resolution in a strong sense.
\begin{definition}\label{def.positive_negative_occurrence}
The notions of \emph{positive} and \emph{negative} subformulas of a formula $A$ are defined as follows. $A$ is a positive subformula of itself. If $B \circ C$ is a positive (resp. negative) subformula of $A$, so are $B$ and $C$, where $\circ \in \{\wedge, \vee\}$. If $\neg B$ is a positive (resp. negative) subformula of $A$, then $B$ is a negative (resp. positive) subformula of $A$. A formula $B$ is a \emph{subformula} of the formula $A$ if it is a positive or negative subformula of $A$.
Positive and negative occurrences of formulas in $\Gamma \Rightarrow \Delta$ are defined as positive and negative occurrences of formulas in $\neg \bigwedge \Gamma \vee \bigvee \Delta$.
\end{definition}
\begin{proposition}\label{prop.Maehara_cutfree_incomplete}
Maehara interpolation in $\LKminus$ is not complete. 
\end{proposition}
\begin{proof}

We use the subformula property of cut-free proofs (see Proposition 4.2.1 and its Corollary in \cite{Troelstra}). Any cut-free proof of $\pi: p \wedge q \Rightarrow p \vee q$ has the following properties:

\begin{itemize}
\item 
For any sequent $\Gamma \Rightarrow \Delta$ in $\pi$, the formulas of $\Gamma$ occur positively in $p \wedge q$ and the formulas of $\Delta$ occur positively in $p \vee q$.
 
\item 
The rules $(L \neg)$ and $(R \neg)$ are not used in $\pi$.

\item 
As the formula $p \wedge q$ occurs negatively in $p \wedge q \Rightarrow p \vee q$ then $p \wedge q$ is only obtained by using left rules, i.e., either $(L \wedge)$, $(L c)$, or $(L w)$. And the rule $(L \vee)$ cannot be used.
 
\item 
Similarly, for $p \vee q$, it can be only obtained using the right rules $(R \vee)$, $(R c)$, or $(R w)$. And the rule $(R \wedge)$ cannot be used.

\item 
Axioms in $\pi$ have one of the forms: $p \Rightarrow p$ or $q \Rightarrow q$. The axioms $\bot \Rightarrow$, $p \wedge q \Rightarrow p \wedge q$ or $p \vee q \Rightarrow p \vee q$ cannot appear in $\pi$, because of the subformula property and the positive or negative occurrences.
\end{itemize}
Thus, the axioms of $\pi$ are either $p\seq p$ or $q\seq q$.
And, as all the rules in $\pi$ are one-premise rules, the interpolants $p \wedge q$ or $p \vee q$ can never be obtained.
\end{proof}

\begin{proposition}
Standard interpolation 
in propositional resolution is not complete.
\end{proposition}
\begin{proof}
The formula $p\land q \impl p\lor q$ has the interpolants $p$, $q$, $p\land q$, or $p\lor q$. However, neither $p\land q$ nor $p\lor q$ can be read off from a resolution proof. To form a resolution refutation, as $p \wedge q \to p \vee q$ is classically valid, the set of clauses $A_1=\{p\}$, $A_2=\{q\}$, $B_1=\{\neg p\}$, and $B_2=\{\neg q\}$ is unsatisfiable. The set of atoms common in both $A_1, A_2$ and $B_1, B_2$ is $\{p, q\}$. A resolution refutation of the mentioned set of clauses has either of the following forms:
\begin{center}
    \begin{tabular}{c c}
\AxiomC{$p$}
\AxiomC{$\neg p$}
\BinaryInfC{$\emptyset$}
\DisplayProof
&
\AxiomC{$q$}
\AxiomC{$\neg q$}
\BinaryInfC{$\emptyset$}
\DisplayProof
    \end{tabular}
\end{center}
Then, based on the algorithm we assign the constant $\bot$ on clauses $A_1$ and $A_2$, and assign $\top$ on clauses $B_1$ and $B_2$.
Then for the left refutation we obtain $sel(p, \bot, \top)$ which is logically
 equivalent to $p$ and for the right refutation we obtain $sel(q, \bot, \top)$
 which is logically equivalent to $q$ on the conclusion as the interpolant in
 each case.  
\end{proof}
Note that the above proofs are quite general.
In the case of resolution, it shows the incompleteness of any
 interpolation algorithm which, when given a resolution refutation in some
 set of propositional variables $V$ will output a formula in $V$.
This is arguably the case for any reasonable interpolation algorithm.
In the case of sequent calculus, it shows the incompleteness of any 
 interpolation algorithm that produces only such interpolants that contain
 only atoms that occur in axioms of the input proof.

\begin{remark}
It is worth mentioning that the above example does not prove the incompleteness of the system resolution with weakening. Because for the same set of clauses $A_1=\{p\}$, $A_2=\{q\}$, $B_1=\{\neg p\}$, and $B_2=\{\neg q\}$, to form the interpolant $p \vee q$ we can take the following resolution refutation:
\begin{center}
\begin{tabular}{c c c}
\AxiomC{$p$}
\RightLabel{\small$(w)$}
\UnaryInfC{$p,q$}
\AxiomC{$\neg p$}
\BinaryInfC{$q$}
\AxiomC{$\neg q$}
\BinaryInfC{$\emptyset$}
\DisplayProof
&
$\overset{\text{interpolant}}{\longrightarrow}$
&
\AxiomC{$\bot$}
\UnaryInfC{$\bot$}
\AxiomC{$\top$}
\BinaryInfC{$sel(p, \bot, \top)=p$}
\AxiomC{$\top$}
\BinaryInfC{$sel(q,p,\top)= p \vee q$}
\DisplayProof
\end{tabular}
\end{center}
\end{remark}

\begin{question}
Is standard interpolation in resolution with weakening complete?    
\end{question}

\begin{question}
Are the standard interpolation algorithms in algebraic proof systems, such as the cutting planes system~\cite{pudlak}, complete?
\end{question}

\section{Interpolants in the sequent calculus}

In this section we prove some general results about interpolants in the
 sequent calculus in preparation for our main results in Sections~\ref{sec.comp_prun_subs} and~\ref{sec.LKat}.

For some aspects of the following proofs, it will be useful
 to distinguish between different occurrences of a formula
 in an {\LK} proof.
We use lowercase Greek letters like $\mu, \nu, \ldots$ to
 denote {\em formula occurrences}.
There is a natural ancestor relation on the set of formula 
 occurrences in a proof: if a formula occurrence $\mu$ is the main
 occurrence of a logical or structural inference rule and $\nu$ is an auxiliary occurrence then $\nu$ is a {\em direct ancestor} of $\mu$.
Moreover, if $\mu$ is a formula occurrence in the context of the
 conclusion sequent of an inference rule and $\nu$ is a corresponding formula occurrence in the context of a premise
 sequent, then $\nu$ is a {\em direct ancestor} of $\mu$.
The {\em ancestor} relation is then the reflexive and transitive closure
 of the direct ancestor relation.
We write $A_\mu$ in a proof to denote an occurrence $\mu$ of a formula $A$.
We also write $L_i$ for a label $L$ of an inference rule to give this inference
 the name $i$.
 \begin{example}
 In the proof
 \[
 \infer[(R\neg)]{A\lor B \seq (\neg (\neg A \land \neg B))_\nu}{
   \infer[(L\lor)_i]{\neg A \land \neg B, (A\lor B)_\mu \seq}{
     \infer[(L\land_1)]{A_{\mu_1}, \neg A \land \neg B \seq}{
       \infer[(L\neg)]{A, \neg A \seq}{
         A\seq A_{\nu_1}
       }
     }
     &
       \infer[(L\land_2)]{B_{\mu_2}, \neg A \land \neg B \seq}{
         \infer[(L\neg)]{B, \neg B \seq}{
           B\seq B_{\nu_2}
         }
       }     
   }
 }
 \]
 the formula ocurrence $\mu$ has two direct ancestors: $\mu_1$ and $\mu_2$.
 The occurrences $\mu$, $\mu_1$, and $\mu_2$ are the active formulas of
  the inference $i$.
 The formula occurrece $\nu$ has eight ancestors, including itself and the
  formula occurrences $\nu_1$ and $\nu_2$.
\end{example}

We will often work with formulas in conjunctive normal form as a convenient
 representative for a class of formulas up to logical equivalence.
Many of the following results do not depend strongly on the shape and could be adapted to other normal forms. Let $L$ be a set of atoms. By $\calA$ is a clause set in the language $L$ we mean every literal in $\calA$ is either an atom or negation of an atom in $L$.
If $\mathcal{C}$ and $\mathcal{D}$ are clause sets, define $\mathcal{C} \times \mathcal{D} := \{C \cup D \mid C \in \mathcal{C} \ \text{and} \ D \in \mathcal{D}\}$.
\begin{definition} \label{def: CNF}
We define the function $\CNF$ which maps formulas to clause sets
 recursively: $\mathrm{CNF} (\top)= \emptyset$, $\CNF(\bot) = \{\emptyset\}$, $\CNF(\ell) =\{\ell\}$,
$\CNF(A \wedge B) = \CNF(A) \cup \CNF(B)$, and $\CNF(A \vee B) = \CNF(A) \times \CNF(B)$, where $\ell$ is a literal and $A$ and $B$ are formulas. 
\end{definition}

\begin{restatable}{observation}{ObsCNF}
\label{obs.CNF}
Let $A, B, C$ be formulas, let $\ell$ be a literal, and let
 $\circ \in \{ \land, \lor \}$. Then:
\begin{enumerate}
\item\label{obs.CNF.logeq} $\CNF(A)$ is logically equivalent to $A$
\item\label{obs.CNF.comm} $\CNF(A\circ B) = \CNF(B \circ A)$
\item\label{obs.CNF.assoc} $\CNF((A\circ B) \circ C) = \CNF(A \circ (B \circ C))$
\item\label{obs.CNF.andidem} $\CNF(A \land A) = \CNF(A)$
\item\label{obs.CNF.oridem} $\CNF(\ell \lor \ell) = \CNF(\ell)$
\item\label{obs.CNF.andunit} $\CNF(A \land \top) = \CNF(A)$
\item\label{obs.CNF.orunit} $\CNF(A \lor \bot) = \CNF(A)$
\end{enumerate}
\end{restatable}
\begin{proof}[Proof in the appendix.]
\end{proof}

In the upcoming Obervation~\ref{observation} and 
 Lemma~\ref{lem.interpolant_for_clause} we prove two results about partitions
 of the end-sequent which determine the interpolant independently of the proof.
\begin{observation}\label{observation}
Let $\Gamma$ and $\Delta$ be multisets of formulas.
\begin{enumerate}
\item\label{observation.left}
If $\pi$ is an {\LKminus} proof of  $\Gamma ; \Rightarrow \Delta ;$ then
 $\CNF(\calM(\pi)) = \{ \emptyset \}$.
\item\label{observation.right}
If $\sigma$ is an {\LKminus} proof of $ ; \Gamma \Rightarrow ; \Delta $ then 
 $\CNF(\calM(\pi)) = \emptyset$.
\end{enumerate}
\end{observation}
\begin{proof}
In case~(\ref{observation.left}) all inferences work on the left-hand side
 of the Maehara partition.
Therefore all axioms have interpolant $\bot$ and binary inferences induce
 disjunctions.
Unary inferences do not modify the interpolant.
Therefore $\CNF(\calM(\pi)) = \CNF(\bot \lor \cdots \bot) = \{ \emptyset \}$.
Case~(\ref{observation.right}) is symmetric.
\end{proof}

\begin{lemma}\label{lem.interpolant_for_clause}
Let $A$ be a formula, let $\{\ell_1, \ldots, \ell_n \}$ be a non-tautological
 clause, and let $\pi$ be an {\LKminus} proof of
 $A ; \Rightarrow ; \ell_1, \ldots,\ell_n$.
Then we have $\CNF(\calM(\pi)) = \{ M \}$ for some clause $M$ with $M \subseteq \{ \ell_1, \ldots, \ell_n \}$.
\end{lemma}
\begin{proof}
First note that a formula occurrence in $\pi$ is an ancestor
 of $A$ (resp. of the $\ell_i$'s) iff it is on the left-hand side (resp.
 right-hand side) of the Maehara partition.
All binary inferences in $\pi$ operate on ancestors of $A$ and thus induce
 disjunctions in the computation of $\calM(\pi)$.
Therefore, $\calM(\pi) = \bigvee_{S\in\calS} \calM(S)$ 
 where $\calS$ is the set of initial split sequents in $\pi$.
We make a case distinction on the form of an $S\in\calS$.
\begin{enumerate}
\item If $S$ is $B ; \overset{\bot}{\Longrightarrow} B ;$
then the interpolant is $\bot$.
\item The case of $S$ being $; B \overset{\top}{\Longrightarrow} ; B$
is impossible because it would
 entail that both occurrences of $B$ are ancestors of the $\ell_i$'s and hence
 $\{ \ell_1, \ldots, \ell_n \}$ would be tautological.
\item If $S$ is $B ; \overset{B}{\Longrightarrow} ; B$
then the occurrence of $B$ on the succedent
of the sequent is ancestor of $\ell_j$ for some $j\in\{1,\ldots,n\}$ and
 thus the interpolant is $B=\ell_j$.
\item If $S$ is $; B \overset{\neg B}{\Longrightarrow} B ;$ 
then the occurrence of $B$ on the antecedent
 of the sequent is ancestor of $\ell_j$ for some $j\in\{1,\ldots,n\}$
  and thus the interpolant is $\neg B = \ell_j$.
\item If $S$ is $\bot ; \overset{\bot}{\Longrightarrow} $
then the interpolant is $\bot$.
\item The case of $S$ being $; \bot \overset{\top}{\Longrightarrow} $
is impossible because it would entail
 that $\ell_j = \neg \bot$ for some $j\in\{1,\ldots,n\}$ and hence
  $\{ \ell_1, \ldots, \ell_n \}$ would be tautological.
\end{enumerate}
We have thus shown that $\calM(S) \in \{ \bot, \ell_1, \ldots, \ell_n \}$ for all 
 $S\in\calS$.
So, by Observation~\ref{obs.CNF}/(\ref{obs.CNF.comm}),(\ref{obs.CNF.assoc}),(\ref{obs.CNF.oridem}),(\ref{obs.CNF.orunit}), we have $\CNF(\calM(\pi))
 = \{ M \}$ for some $M\subseteq \{ \ell_1,\ldots,\ell_n \}.$
\end{proof}

The next useful result is that $\calM$ is just as complete in {\LKat} as it
 is in {\LKlit}.
To show this, we first need a version of the inversion lemma for negation 
 that preserves the interpolant.
\begin{restatable}{lemma}{LemNegInversion}
\label{lem.neginversion}
If $\pi$ is an {\LKm} proof with monochrome cuts of
\begin{enumerate}
\item $\Mseq{\Gamma_1}{\Gamma_2}{\Delta_1}{\Delta_2, \neg A}$
\item $\Mseq{\Gamma_1}{\Gamma_2}{\Delta_1, \neg A}{\Delta_2}$
\item $\Mseq{\Gamma_1}{\Gamma_2, \neg A}{\Delta_1}{\Delta_2}$
\item $\Mseq{\Gamma_1, \neg A}{\Gamma_2}{\Delta_1}{\Delta_2}$
\end{enumerate}
then there is an {\LKm} proof with monochrome cuts $\pi'$ of
\begin{enumerate}
\item $\Mseq{\Gamma_1}{\Gamma_2, A}{\Delta_1}{\Delta_2}$
\item $\Mseq{\Gamma_1, A}{\Gamma_2}{\Delta_1}{\Delta_2}$
\item $\Mseq{\Gamma_1}{\Gamma_2}{\Delta_1}{\Delta_2, A}$
\item $\Mseq{\Gamma_1}{\Gamma_2}{\Delta_1, A}{\Delta_2}$
\end{enumerate}
with $\calM(\pi') = \calM(\pi)$ and $|\pi'|  \leq 2|\pi|$.
\end{restatable}
\begin{proof}[Proof in the appendix]
\end{proof}

\begin{lemma}\label{lem.LKatLKlitEquivInt}
If $\pi$ is an {\LKlit} proof of
 $\Mseq{\Gamma_1}{\Gamma_2}{\Delta_1}{\Delta_2}$
 then there is an {\LKat} proof $\pi'$ of
 $\Mseq{\Gamma_1}{\Gamma_2}{\Delta_1}{\Delta_2}$ with
 $\CNF(\calM(\pi')) = \CNF(\calM(\pi))$ and $|\pi'| \leq 2|\pi|$. 
\end{lemma}
\begin{proof}
As all the cuts in $\pi$ are on literals, we apply Lemma~\ref{lem.neginversion} to each instance of a cut rule on negative literals in the proof to obtain a proof with only atomic cuts. Take a topmost instance of a cut rule where the cut formula is a negative literal. Call this subproof $\sigma$. Based on the position of the cut formula in the Maehara partition, $\sigma$ either looks like
\begin{prooftree}
\AxiomC{$\sigma_1$}
\noLine
 \UnaryInfC{$\Gamma_1; \Gamma_2 \overset{I}{\Longrightarrow} \Delta_1, \neg p ; \Delta_2$}
\AxiomC{$\sigma_2$}
\noLine
 \UnaryInfC{$\Gamma_1, \neg p; \Gamma_2 \overset{J}{\Longrightarrow} \Delta_1 ; \Delta_2$}
 \BinaryInfC{$\Gamma_1; \Gamma_2 \overset{I \vee J}{\Longrightarrow} \Delta_1 ; \Delta_2$}    
\end{prooftree}
or
\begin{prooftree}
\AxiomC{$\sigma_1$}
\noLine
 \UnaryInfC{$\Gamma_1; \Gamma_2 \overset{I}{\Longrightarrow} \Delta_1 ; \neg p, \Delta_2$}
\AxiomC{$\sigma_2$}
\noLine
 \UnaryInfC{$\Gamma_1 ; \Gamma_2, \neg p \overset{J}{\Longrightarrow} \Delta_1 ; \Delta_2$}
 \BinaryInfC{$\Gamma_1; \Gamma_2 \overset{I \wedge J}{\Longrightarrow} \Delta_1 ; \Delta_2$}    
\end{prooftree} 
where $p$ is an atom and the cut formula is $\neg p$.
In the first case, we apply Lemma \ref{lem.neginversion} to $\sigma_1$ and $\sigma_2$ to get $\sigma_1'$ and $\sigma_2'$ and use the cut to get the end sequent of the proof $\sigma$:
\begin{prooftree}
\AxiomC{$\sigma_1'$}
\noLine
 \UnaryInfC{$\Gamma_1, p; \Gamma_2 \overset{I}{\Longrightarrow} \Delta_1 ; \Delta_2$}
\AxiomC{$\sigma_2'$}
\noLine
 \UnaryInfC{$\Gamma_1; \Gamma_2 \overset{J}{\Longrightarrow} \Delta_1 , p; \Delta_2$}
 \BinaryInfC{$\Gamma_1; \Gamma_2 \overset{I \vee J}{\Longrightarrow} \Delta_1 ; \Delta_2$}    
\end{prooftree}
Now, this cut is atomic, the cut formula is $p$, and the interpolant remains
 the same up to  commutativity of conjunction. The other case is similar.
We apply the same process to every cut on a negated literal in $\pi$, resulting in a proof with only atomic cuts.
\end{proof}

The main technical lemma of Section~\ref{sec.comp_prun_subs},  
 Lemma~\ref{lem.ce_ipol}, will be shown by carrying out a cut-elimination
 argument on a carefully chosen class of proofs. This class on
 the one hand is large enough to permit an embedding of all pruned
 interpolants, but on the other hand small enough to exhibit a very nice
 behaviour during cut-elimination: the interpolant of the reduced proof is
 subsumed by the interpolant of the original proof.
We now proceed to introduce this class of proofs, called ``tame" proofs,
 which is a new invariant for cut-elimination.

\begin{definition}
We say that a cut is of type R if it is of the form
\[
\infer[\mathrm{cut}]{\Mseq{\Gamma_1}{\Gamma_2}{\Delta_1}{\Delta_2}}{
  \Mseq{\Gamma_1}{\Gamma_2}{\Delta_1}{\Delta_2, C}
  &
  \Mseq{\Gamma_1}{\Gamma_2, C}{\Delta_1}{\Delta_2}
}
\]
and of type L if it is of the form
\[
\infer[\mathrm{cut}]{\Mseq{\Gamma_1}{\Gamma_2}{\Delta_1}{\Delta_2}}{
  \Mseq{\Gamma_1}{\Gamma_2}{\Delta_1, C}{\Delta_2}
  &
  \Mseq{\Gamma_1, C}{\Gamma_2}{\Delta_1}{\Delta_2}
}
\]
\end{definition}

Our cut-elimination argument will work with proofs all of whose cuts
 are of type $R$.
\begin{definition}\label{def: LR}
We say that an axiom is of type L/L if it is of the form
$\Mseq{A}{}{A}{}$, of type L/R if it is of the form $\Mseq{A}{}{}{A}$,
of type R/L if it is of the form $\Mseq{}{A}{A}{}$, and of type R/R
if it is of the form $\Mseq{}{A}{}{A}$.
\end{definition}

\begin{definition}
We say that an axiom occurrence $A\seq A$ in a proof $\pi$ is of type 
 $\Omega$ if both occurrences of $A$ are ancestors of cut formulas in
 $\pi$.
\end{definition}

Every cut $c$ in a proof has a {\em left subproof} and a {\em right subproof}:
 the subproof whose end-sequent is the left, or respectively: right, premise
 sequent of $c$.
\begin{definition}\label{def.tame}
An {\LKm} proof $\pi$ is called {\em tame} if
\begin{enumerate}
\item\label{def.tame.Omega} $\pi$ does not contain axioms of type $\Omega$ and
\item\label{def.tame.RR} every cut in $\pi$ has a subproof
 in which all axioms in which an ancestor of the cut formula is active
 are of type R/R.
\end{enumerate}
\end{definition}

\begin{definition}
A clause set $\calA$ is called {\em pruned} if no atom occurs both positively
 and negatively in $\calA$ and $\calA$ does not contain the literal $\top$.
\end{definition}
For instance, none of the following clause sets are pruned: 
\[\{\{p\}, \{r, \neg p\}\} \qquad \{\{\top, p\}\} \qquad \{\{p, \neg p\}, \{r\}\}
\]
%

\begin{restatable}{lemma}{LemPruneClauseSet}
\label{lem.prune_clause_set}
Let $\calA$ be a clause set in some language $L$.
Let 
\[
L_\mathrm{D} = \{ p \in L \mid \text{$p$ occurs both positively and negatively
 in $\calA$} \}.
 \]
Then, there is a pruned clause set $\calA^*$ in the language $L\setminus L_\mathrm{D}$
 such that
\begin{enumerate}
\item \label{prune1} $I\models \calA$ implies $I \models \calA^*$.
\item \label{prune2} 
$I' \models \calA^*$ implies that there is an extension of $I'$ to an  $L$ interpretation $I$ such that $I \models \calA$.
\end{enumerate}
\end{restatable}
\begin{proof}[Proof in the appendix]
\end{proof}
The above lemma is shown essentially by computing the closure $\calA'$
 of $\calA$ under resolution and then obtaining $\calA^*$ from $\calA'$
 by deleting all clauses that contain an atom in $L_\mathrm{D}$, including
 all tautologies.
Thus, pruned clause sets allow for a simple standardised representation of
 a formula or a clause set.

The following useful Lemma essentially combines existing interpolants with
 a conjunction.
\begin{lemma}\label{lem.interpolant_conjunction_construction}
Let $A \impl B$ be a valid formula, let the clause set
 $\calC=\{C_i \mid 1 \leq i \leq n\}$ be an interpolant of $A\impl B$.
For $i=1,\ldots,n$ let $\pi_i: \Mseq{A}{}{}{C_i}$ be an {\LKat} proof.
Then, there is an {\LKat} proof $\psi$ of $\Mseq{A}{}{}{B}$ 
 all of whose cuts are of type R with $\CNF(\calM(\psi))=\CNF(\bigwedge_{i=1}^n \calM(\pi_i))$.
Moreover, if $\calC$ is a pruned clause set,  then $\psi$ is tame.
\end{lemma}

\begin{proof}
As $\calC=\bigwedge_{i=1}^n \bigvee_{j=1}^{k_i} \ell_{ij}$ is an interpolant of
 $A \to B$, both $A \Rightarrow \calC$ and $\calC \Rightarrow B$ are valid sequents.
Take $\neg \calC=\neg \bigwedge_{i=1}^n \bigvee_{j=1}^{k_i} \ell_{ij}$ which is 
 logically equivalent to $\bigvee_{i=1}^{n} \bigwedge_{j=1}^{k_i} 
 \overline{\ell_{ij}}$ and to the clause set $\CNF(\bigvee_{i=1}^{n} \bigwedge_{j=1}^{k_i}\overline{\ell_{ij}})$.
Define

\begin{center}
$ \begin{cases}
\calC^+= \{\{\ell_{i1}, \dots, \ell_{ik_i}\} \mid 1 \leq i \leq n\} \\ 
\calC^-= \{\{\overline{\ell_{1j_1}}, \dots, \overline{\ell_{nj_n}}\} \mid 1 \leq j_i \leq k_i\} 
\end{cases} $
\end{center}

The set of clauses $\calC^+ \cup \calC^-$ is unsatisfiable. 
Let $F$ be a resolution refutation with these clauses as the initial clauses and 
 $\emptyset$ as the conclusion of the refutation.
Transform $F$ to a split proof in {\LKlit} as follows.
Take the initial clauses; if a clause is in $\calC^+$ then it is of the form
 $C_i=\{\ell_{i1}, \dots, \ell_{ik_i}\}$ for some $1 \leq i \leq n$.
And if an initial clause is in $\calC^-$, then it is of the form
 $D_{j_i}= \{ \overline{\ell_{1j_1}}, \dots, \overline{\ell_{nj_n}}\}$ for
  some $1 \leq j_i \leq k_i$.
Replace these clauses with the following split sequents:
\[
C_i \;\; \text{with} \;\; A; \Rightarrow ; \ell_{i1}, \dots, \ell_{ik_i} \quad D_{j_i} \;\; \text{with} \;\; ; \ell_{1j_1}, \dots, \ell_{nj_n} \Rightarrow ; B
\]
Now, suppose a resolution rule is applied on a literal $\ell$ in $F$
\begin{prooftree}
 \AxiomC{$M, \ell$} 
  \AxiomC{$\overline{\ell}, N$} 
  \BinaryInfC{$M , N$}
\end{prooftree}
We replace this rule with a cut on the literal $\ell$ where the premises of the
 cut rule are the corresponding split sequents of $M, \ell$ and $\ell, N$.
After making all these replacements, we obtain a derivation $\psi_1$ of
 $A ; \Rightarrow ;B$, where the
 split sequents $A; \Rightarrow ; \ell_{i1}, \dots, \ell_{ik_i}$ and $ ; \ell_{1j_1}, \dots, \ell_{nj_n} \Rightarrow ; B$, for all $1 \leq i \leq n$ and $1 \leq j_i \leq k_i$ appear as leaves of $\psi_1$. 
Note that in $\psi_1$ inferences of the following form appear:
\begin{center}
\AxiomC{$\pi_1$}
\noLine
\UnaryInfC{$\Gamma_1; \Gamma_2 \Rightarrow \Delta_1; \ell, \Delta_2$}
\AxiomC{$\pi_2$}
\noLine
\UnaryInfC{$\Sigma_1; \ell, \Sigma_2 \Rightarrow \Lambda_1; \Lambda_2$}
\BinaryInfC{$\Gamma_1, \Sigma_1; \Gamma_2, \Sigma_2 \Rightarrow \Delta_1, \Lambda_1; \Delta_2, \Lambda_2$} 
\DisplayProof
\end{center}
for some literal $\ell$. We use these inferences as abbreviations for 
\small\begin{center}
\AxiomC{$\pi_1$}
\noLine
\UnaryInfC{$\Gamma_1; \Gamma_2 \Rightarrow \Delta_1; \ell, \Delta_2$}
\doubleLine
\UnaryInfC{$\Gamma_1, \Sigma_1; \Gamma_2, \Sigma_2 \Rightarrow \Delta_1, \Lambda_1; \ell, \Delta_2, \Lambda_2$}
\AxiomC{$\pi_2$}
\noLine
\UnaryInfC{$\Sigma_1; \ell, \Sigma_2 \Rightarrow \Lambda_1; \Lambda_2$}
\doubleLine
\UnaryInfC{$\Gamma_1, \Sigma_1; \ell, \Gamma_2, \Sigma_2 \Rightarrow \Delta_1, \Lambda_1; \Delta_2, \Lambda_2$}
\BinaryInfC{$\Gamma_1, \Sigma_1; \Gamma_2, \Sigma_2 \Rightarrow \Delta_1, \Lambda_1; \Delta_2, \Lambda_2$} 
\DisplayProof
\end{center}

\normalsize \noindent where the double lines mean applying the left and right weakening rules as often as needed. For the curious reader, this means that we can replace each \emph{context-splitting} cut (a cut rule of the first form) with a combination of weakening rules and \emph{context-sharing} cut rules, to get a derivation in $\LKlit$.
Now, replace each leaf of $\psi_1$ of the form  $A ; \Rightarrow ; \ell_{i1}, \dots, \ell_{ik_i}$
 by the proof $\pi_i$. 
Moreover, replace each leaf of $\psi_1$ of the form  $ ; \ell_{1j_1}, \dots, \ell_{nj_n} \Rightarrow ; B$ by a cut-free proof $\sigma_{j_1,\ldots,j_n}$ of it which exists
 since it is a valid sequent. 
The result of making all these replacements yields an {\LKlit} proof $\psi_2: A; \Rightarrow ; B$. 
Now, we investigate the Maehara interpolant of $\psi_2$. Note that in each cut rule, the cut formula is on the right-hand side of the Maehara partition. The initial sequents of $\psi_2$ are either the initial sequents of $\pi_i$'s or initial sequents of the cut-free proofs of the sequents 
$; \ell_{1j_1}, \dots, \ell_{nj_n} \Rightarrow ; B$.
For the latter, by Observation \ref{observation}, the interpolant is $\top$.
On the other hand, the proof $\pi_i$ has the interpolant $\calM(\pi_i)$.
Hence, $\calM(\psi_2)$ is the conjunction of $\bigwedge_{i=1}^n \calM(\pi_i)$ with several $\top$'s.
The final step is  applying Lemma \ref{lem.LKatLKlitEquivInt} to obtain a proof $\psi: A; \Rightarrow ; B$ in $\LKat$. We have
\[
\calM(\psi_1)=\calM(\psi_2)=\calM(\psi)
\]
and we have $\CNF(\calM(\psi))= \CNF(\bigwedge_{i=1}^n \calM(\pi_i))$.
All cuts in the proofs $\psi_1$, $\psi_2$, and $\psi$ are of type R.
Since $\calC$ is a pruned clause set, no $C_i$ is a tautology, so no $\pi_i$
 contains an axiom of type $\Omega$.
Similarly, since $\calC$ is a pruned clause set, no $\{ \overline{\ell_{1,j_1}}, \ldots, \overline{\ell_{n,j_n}} \}$ contains two dual literals, so no
 $\sigma_{j_1,\ldots,j_n}$ contains an axiom of type $\Omega$.
Therefore tameness condition~(\ref{def.tame.Omega}) is satisfied.
Since all formula occurrences in $\sigma_{j_1,\ldots,j_n}$ are on the right-hand
 side of the split sequents, all axioms on the right-hand side of a cut on
  $\ell_{i,j_i}$ in which an ancestor of the cut formula is active are of type R/R.
Therefore tameness condition~(\ref{def.tame.RR}) is satisfied.

\end{proof}

\begin{remark}
Up to associativity, commutativity, idempotence, and unit elimination of $\wedge$ we have $\calM(\psi)=\bigwedge_{i=1}^n \calM(\pi_i)$ in Lemma \ref{lem.interpolant_conjunction_construction}. 
\end{remark}

\section{Completeness up to pruning and subsumption}\label{sec.comp_prun_subs}

Even though, as we have seen in 
 Proposition~\ref{prop.Maehara_cutfree_incomplete}, Maehara interpolation
 in {\LKminus} is incomplete it is possible to obtain also positive
 results for {\LKminus}.
In this section we will prove such a positive result: we will show that, if we
 restrict our attention to {\em pruned interpolants}, then
 Maehara interpolation is complete up to subsumption.
The proof strategy consists in carrying out a cut-elimination procedure
 on tame proofs with monochromatic cuts and showing that, in this setting,
 the interpolants have a very nice behaviour: the interpolant of the reduced
 proof is subsumed by the interpolant of the original proof.
Applying this cut-elimination on a suitable chosen proof with atomic cuts
 will yield the desired result.

\begin{definition}
A clause set $\calA$ {\em subsumes} a clause set $\calB$, in symbols
 $\calA \subs \calB$, if for all $B \in \calB$ there is an $A \in \calA$
 s.t.\ $A\subseteq B$.
\end{definition}
For instance, $\{\{p\}\}$ subsumes $\{\{p,q\}, \{p\}\}$.
Subsumption is one of the most useful and one of the most thoroughly studied
 mechanisms for the detection and
 elimination of redundancy in automated deduction.
Note that, if $\calA \subs \calB$ then $\calA \models \calB$.
In this sense, subsumption is a restricted form of implication.

\begin{restatable}{observation}{ObsSubs}
\label{obs.subs}
Let $\calA, \calB, \calC$ be clause sets.
\begin{enumerate}
\item\label{obs.subs.subset} If $\calA \supseteq \calB$ then $\calA \subs \calB$.
\item\label{obs.subs.trans} If $\calA \subs \calB$ and $\calB \subs \calC$
  then $\calA \subs \calC$.
\item\label{obs.subs.compat_union} If $\calA \subs \calB$ then $\calA \cup \calC \subs \calB \cup \calC$.
\item\label{obs.subs.compat_times} If $\calA \subs \calB$ then $\calA \times \calC \subs \calB \times \calC$.
\item\label{obs.subs.distand} $(\calA \times \calB) \cup \calC \subs (\calA \cup \calC)\times (\calB \cup \calC)$
\end{enumerate}
\end{restatable}
\begin{proof}[Proof in the appendix.]
\end{proof}

We proceed to set up the initial {\LKat} proof for our
 cut-elimination argument.
\begin{definition}
A pruned clause set $\calC$ is called {\em pruned interpolant} of a formula
 $A\impl B$ if there are no $C' \subset C \in \calC$ with $A\models C'$.
\end{definition}
So a pruned interpolant, in addition to being a pruned clause set,
 must not contain redundant literals in the sense of the above definition.
 
\begin{lemma}\label{lem.printerpolant_LKatproof}
Let $\calC$ be a pruned interpolant of an implication $A\impl B$.
Then there is a tame {\LKat} proof $\pi$ of $\Mseq{A}{}{}{B}$
 all of whose cuts are of type R with $\CNF(\calM(\pi)) = \calC$.
\end{lemma}
\begin{proof}
Let $\calC = \{ C_i \mid 1 \leq i \leq n \}$.
Since $\calC$ is an interpolant, $\models A\impl C_i$ for all $i=1,\ldots,n$.
Since $\calC$ is pruned, the $C_i$ are non-tautological, so
 Lemma~\ref{lem.interpolant_for_clause} yields an {\LKminus}
 proof $\pi_i$ of $\Mseq{A}{}{}{C_i}$ with $\CNF(\calM(\pi_i)) = \{ M_i \}$
 for some $M_i \subseteq C_i$.
Since $\calM(\pi_i)$ is an interpolant of $A\impl C$, we have
 $\models A \impl \calM(\pi_i)$.
By Observation~\ref{obs.CNF}/(\ref{obs.CNF.logeq}), we have
 $\models A \impl \CNF(\calM(\pi_i))$, i.e.,
 $\models A \impl M_i$.
Then $M_i = C_i$ because $M_i \subset C_i$ would contradict prunedness
 of $\calC$.
Then, by applying Lemma~\ref{lem.interpolant_conjunction_construction},
 we obtain a tame {\LKat} proof $\pi$ all of whose cuts are of type R with
 $\CNF(\calM(\pi)) = \CNF(\Land_{i=1}^n \calM(\pi_i)) = \calC$.
\end{proof}

\begin{definition}
We call a proof $\pi$ {\em w-reduced} if every weakening inference
 in $\pi$ occurs immediately below an axiom or another weakening
 inference.
\end{definition}
The point of the notion of w-reduced proofs is to facilitate the
 technical matters of the cut-elimination argument in our
 variant of {\LK}.
We now proceed to set up the termination measure for the
 cut-elimination procedure.
\begin{definition}
A formula occurrence in a proof $\pi$ is called {\em weak} if all its
 ancestors are introduced by weakening inferences.

Let $\mu$ be an occurrence of a formula $A$ in a proof $\pi$.
An occurrence $\mu'$ of $A$ is called {\em relevant} for $\mu$
 if $\mu'$ is an ancestor of $\mu$, $\mu'$ is not weak, and
 $\mu'$ is not in the conclusion sequent of a weakening inference.

The \emph{weight of a formula occurrence} $\mu$ in a proof $\pi$, written
 as $\weight(\mu)$, is the number of formula occurrences which are 
 relevant for $\mu$.

The \emph{weight of a cut $c$} in a proof $\pi$ is defined as
 $\weight(c) = \weight(\mu_\mathrm{l}) + \weight(\mu_\mathrm{r})$
 where $\mu_\mathrm{l}$ (resp. $\mu_\mathrm{r}$) is the occurrence of
 the cut formula in the left (resp. right) premise of $c$.
\end{definition}

\begin{definition}
The {\em degree of a cut $c$}, written as $\deg(c)$, is the logical
 complexity, i.e., the number of propositional connectives, of
 the cut formula of $c$.
\end{definition}

\begin{restatable}{lemma}{LemWReduce}
\label{lem.wreduce}
Let $G \in \{ \LKm, \LKminus \}$.
For every $G$ proof $\pi$ of a sequent $\Mseq{\Gamma_1}{\Gamma_2}{\Delta_1}{\Delta_2}$ there is a w-reduced $G$ proof $\pi'$ of
 $\Mseq{\Gamma_1}{\Gamma_2}{\Delta_1}{\Delta_2}$ such that 
\begin{enumerate}
\item\label{lem.wreduce.ipol} $\calM(\pi') = \calM(\pi)$,
\item\label{lem.wreduce.tame} if $\pi$ is tame then so is $\pi'$
\item\label{lem.wreduce.Rcut} if all cuts in $\pi$ are of type R then all cuts in $\pi'$ are of type R
\item\label{lem.wreduce.focc} for any formula occurrence $\mu$ in the end-sequent of $\pi$ and
 its corresponding formula occurrence $\mu'$ in the end-sequent
 of $\pi'$:
\begin{enumerate}
\item\label{lem.wreduce.focc.weight} $\weight(\mu') = \weight(\mu)$ and
\item\label{lem.wreduce.focc.ax} every axiom of $\pi'$ in which an ancestor of $\mu'$ is active is,
 up to weak formulas, also an axiom of $\pi$ in which an ancestor of $\mu$
 is active. 
\end{enumerate}
\end{enumerate}
\end{restatable}
\begin{proof}[Proof Sketch.]
Shifting up weakenings until they are in a permitted
 position satisfies the mentioned properties.
For a more detailed proof, please see the appendix.
\end{proof}

\begin{restatable}{lemma}{LemCEIpol}
\label{lem.ce_ipol}
For every tame {\LKm} proof $\pi$ all of whose cuts are of type R,
 there is an {\LKminus} proof $\pi'$ with $\CNF(\calM(\pi)) \subs \CNF(\calM(\pi'))$.
\end{restatable}
\begin{proof}
By Lemma~\ref{lem.wreduce} we can assume that $\pi$ is tame and w-reduced.
In this proof we will write 
\[
\infer[\mathrm{wax}]{\Gamma,A \seq \Delta, A}{}
\]
as an abbreviation for the axiom $A\seq A$ followed by weakening inferences
 to derive $\Gamma, A \seq \Delta, A$.

We write $\pi$ as $\pi[\chi]$ where $\chi$ is a subproof of $\pi$ that ends with an
 uppermost cut.
Based on a case distinction on the form of $\chi$
 we will define a proof $\chi^*$ which, by replacing $\chi$, yields
 a proof $\pi[\chi^*]$.
We will show that
\begin{enumerate}[label=(\roman*)]
\item\label{lem.ce_ipol.tamewred} $\pi[\chi^*]$ is tame and w-reduced,
\item\label{lem.ce_ipol.weight} the cut $c$ in $\chi$ is replaced by cuts $c'$ in $\chi^*$ with $\degree(c) > \degree(c')$ or ( $\degree(c) = \degree(c')$ and $\weight(c) > \weight(c')$ ), and
\item\label{lem.ce_ipol.M} $\calM(\chi) \subs \calM(\chi^*)$.
\end{enumerate}
Points~\ref{lem.ce_ipol.tamewred} and~\ref{lem.ce_ipol.weight} ensure
 correctness and termination of the cut-elimination sequence
 while point~\ref{lem.ce_ipol.M} shows $\CNF(\calM(\pi[\chi])) \subs \CNF(\calM(\pi[\chi^*]))$ by 
 Observation~\ref{obs.subs}/(\ref{obs.subs.compat_union}) and (\ref{obs.subs.compat_times}).
This suffices to prove the result by transitivity of $\subs$.
In order to show that $\pi[\chi^*]$ is tame we will show
 the following condition (*) in most cases:

Every axiom of $\chi^*$ is, up to weak formulas, an axiom of $\chi$.

\noindent Hence tameness condition~(\ref{def.tame.Omega}) is preserved, because $\chi$ is tame.
Moreover, if $\mu$ is a formula occurrence in the end-sequent of $\chi$ and
 $\mu^*$ is the corresponding formula occurrence in the end-sequent of $\chi^*$,
 then every axiom of $\chi^*$ in which an ancestor of
 $\mu^*$ is active is, up to weak formulas, also an axiom of $\chi$ in which $\mu$
 is active.
Therefore, also tameness condition~(\ref{def.tame.RR}) is preserved.

{\bf Exclusion of weak cut formulas:}
If $\chi =$
\[
\infer[\mathrm{cut}]{\Mseq{\Gamma_1}{\Gamma_2}{\Delta_1}{\Delta_2}}{
  \deduce{\Mseq{\Gamma_1}{\Gamma_2}{\Delta_1}{\Delta_2, C_{\mu_1}}}{(\chi_1)}
  &
  \deduce{\Mseq{C_{\mu_2}, \Gamma_1}{\Gamma_2}{\Delta_1}{\Delta_2}}{(\chi_2)}
}
\]
where $\mu_1$ is weak, define $\chi^* : \Mseq{\Gamma_1}{\Gamma_2}{\Delta_1}{\Delta_2}$ from $\chi_1$ by removing $\mu_1$, all its ancestors, and all
 weakening inferences that introduce these ancestors.
Then $\pi[\chi^*]$ is w-reduced and, since $\chi_1, \chi^*$ satisfy
(*), $\pi[\chi^*]$ is also tame which shows~\ref{lem.ce_ipol.tamewred}.
\ref{lem.ce_ipol.weight} holds vacuously.
Furthermore, $\CNF(\calM(\chi)) = \CNF(\calM(\chi_1)\land \calM(\chi_2))
 \supseteq \CNF(\calM(\chi_1)) = \CNF(\calM(\chi^*))$.
So~\ref{lem.ce_ipol.M} follows from Observation~\ref{obs.subs}/(\ref{obs.subs.subset}).
In case $\mu_2$ is weak we proceed analogously.
So, in the remaining cases, we assume that none of the two cut formulas are weak.

{\bf Permutation of a binary inference over a cut:}
If $\chi$ is of the form
\[
\small\infer[\mathrm{cut}_c]{\Mseq{\Gamma_1}{\Gamma_2}{\Delta_1, A\land B}{\Delta_2}}{
  \infer{\Mseq{\Gamma_1}{\Gamma_2}{\Delta_1, A\land B}{\Delta_2, C}}{
    \deduce{\Mseq{\Gamma_1}{\Gamma_2}{\Delta_1, A}{\Delta_2, C_{\mu_1}}}{(\chi_1)}
    &
    \deduce{\Mseq{\Gamma_1}{\Gamma_2}{\Delta_1, B}{\Delta_2, C_{\mu_2}}}{(\chi_2)}
  }
  & \hspace*{-30pt}
  \deduce{\Mseq{\Gamma_1}{\Gamma_2,C_{\mu_3}}{\Delta_1, A\land B}{\Delta_2}}{(\chi_3)}
}
\]
we define the proofs
\begin{align*}
\chi'_1 \quad & = \quad
\begin{array}{c}
\infer{\Mseq{\Gamma_1}{\Gamma_2}{\Delta_1,A,A\land B}{\Delta_2,C}}{
  \deduce{\Mseq{\Gamma_1}{\Gamma_2}{\Delta_1,A}{\Delta_2,C}}{(\chi_1)}
}
\end{array} \\
\chi'_2 \quad & = \quad
\begin{array}{c}
\infer{\Mseq{\Gamma_1}{\Gamma_2}{\Delta_1,B,A\land B}{\Delta_2,C}}{
  \deduce{\Mseq{\Gamma_1}{\Gamma_2}{\Delta_1,B}{\Delta_2,C}}{(\chi_2)}
}
\end{array} \\
\chi'_{3,1} \quad & = \quad
\begin{array}{c}
\infer{\Mseq{\Gamma_1}{\Gamma_2, C}{\Delta_1, A, A\land B}{\Delta_2}}{
  \deduce{\Mseq{\Gamma_1}{\Gamma_2,C}{\Delta_1,A\land B}{\Delta_2}}{(\chi_3)}
}
\end{array} \\
\chi'_{3,2} \quad & = \quad
\begin{array}{c}
\infer{\Mseq{\Gamma_1}{\Gamma_2, C}{\Delta_1,B,A\land B}{\Delta_2}}{
  \deduce{\Mseq{\Gamma_1}{\Gamma_2, C}{\Delta_1,A\land B}{\Delta_2}}{(\chi_3)}
}
\end{array}
\end{align*}
By applying Lemma~\ref{lem.wreduce} to these proofs individually,
 we obtain proofs $\chi^*_1$, $\chi^*_2$, $\chi^*_{3,1}$, and $\chi^*_{3,2}$
 that satisfy~(\ref{lem.wreduce.ipol}) and~(\ref{lem.wreduce.focc}).
We finally define $\chi^* =$
\[
\infer{\Mseq{\Gamma_1}{\Gamma_2}{\Delta_1,A\land B}{\Delta_2}}{
  \infer[(R\land)]{\Mseq{\Gamma_1}{\Gamma_2}{\Delta_1,A\land B,A\land B}{\Delta_2}}{
    \infer[\mathrm{cut}_{c_1}]{\Mseq{\Gamma_1}{\Gamma_2}{\Delta_1,A,A\land B}{\Delta_2}}{
      (\chi^*_1)
      & & & &
      (\chi^*_{3,1})
    }
    &
    \infer[\mathrm{cut}_{c_2}]{\Mseq{\Gamma_1}{\Gamma_2}{\Delta_1,B,A\land B}{\Delta_2}}{
      (\chi^*_2)
      & & & &
      (\chi^*_{3,2})
    }
  }
}
\]
For~\ref{lem.ce_ipol.tamewred} we observe that $\pi[\chi^*]$ is w-reduced
 and, since $\chi,\chi^*$ satisfy (*), $\pi[\chi^*]$ is tame.
%
%
%
For~\ref{lem.ce_ipol.weight} we observe that $\weight(c) = \weight(\mu_1) + w(\mu_2) + 1 + \weight(\mu_3)$, $\weight(c_1) = \weight(\mu_1) + \weight(\mu_3)$, and $\weight(c_2) = \weight(\mu_2) + \weight(\mu_3)$.
For~\ref{lem.ce_ipol.M} we have $\CNF(\calM(\chi))
 = \CNF((\calM(\chi_1) \lor \calM(\chi_2))\land \calM(\chi_3))
 \subs^\text{Obs.~\ref{obs.subs}/(\ref{obs.subs.distand})} \CNF((\calM(\chi_1) \land \calM(\chi_3)) \lor
  (\calM(\chi_2) \land \calM(\chi_3)))
 = \CNF(\calM(\chi^*))$.

We proceed analogously for the other binary inferences.
If the binary inference above the cut works on the right-hand side of the
 Maehara partition, we define
 $\chi^*$ analogously and obtain the calculation
$\CNF(\calM(\chi))
 = \CNF((\calM(\chi_1) \land \calM(\chi_2))\land \calM(\chi_3))
 =^\text{Obs.~\ref{obs.CNF}} \CNF((\calM(\chi_1) \land \calM(\chi_3)) \land
  (\calM(\chi_2) \land \calM(\chi_3)))
 = \CNF(\calM(\chi^*))$.

{\bf Reduction of axioms:}
Since all cuts are of type R, we have to consider the four cases of axioms
 of type $T_1$/R and R/$T_2$ for $T_1,T_2 \in \{\mathrm{L},\mathrm{R}\}$.

If $\chi$ is of the form R/R and R/R as in
\[
\infer[\mathrm{cut}]{\Mseq{\Gamma_1}{\Gamma_2, A_{\mu_1}}{\Delta_1}{\Delta_2, A_{\mu_2}}}{
  \infer[\mathrm{wax}]{\Mseq{\Gamma_1}{\Gamma_2, A}{\Delta_1}{\Delta_2, A, A}}{}
  &
  \infer[\mathrm{wax}]{\Mseq{\Gamma_1}{\Gamma_2, A, A}{\Delta_1}{\Delta_2, A}}{}
}
\]
we reduce $\chi$ to $\chi^* =$
\[
\infer[\mathrm{wax}]{\Mseq{\Gamma_1}{\Gamma_2, A}{\Delta_1}{\Delta_2,A}}{}
\]
For \ref{lem.ce_ipol.tamewred} we observe that $\pi[\chi^*]$ is w-reduced.
Moreover, the axiom in $\chi^*$ is not of type $\Omega$ for if it were,
 then both $\mu_1$ and $\mu_2$ would
 be ancestors of cuts in $\pi$ and hence $\chi$ would contain two axioms
 of type $\Omega$.
\ref{lem.ce_ipol.weight} is trivial.
For \ref{lem.ce_ipol.M} we have
 $\CNF(\calM(\chi)) = \CNF(\top\land\top)
 =^\text{Obs.~\ref{obs.CNF}/(\ref{obs.CNF.andunit})} \CNF(\top)
 = \CNF(\calM(\chi^*))$.

If $\chi$ is of the form R/R and R/L as in
\[
\infer[\mathrm{cut}]{\Mseq{\Gamma_1}{\Gamma_2, A}{\Delta_1, A}{\Delta_2}}{
  \infer[\mathrm{wax}]{\Mseq{\Gamma_1}{\Gamma_2, A}{\Delta_1, A}{\Delta_2, A}}{}
  &
  \infer[\mathrm{wax}]{\Mseq{\Gamma_1}{\Gamma_2, A, A}{\Delta_1, A}{\Delta_2}}{}
}
\]
we reduce $\chi$ to $\chi^* =$
\[
\infer[\mathrm{wax}]{\Mseq{\Gamma_1}{\Gamma_2,A}{\Delta_1,A}{\Delta_2}}{}
\]
\ref{lem.ce_ipol.tamewred} is shown as above.
\ref{lem.ce_ipol.weight} is trivial.
For \ref{lem.ce_ipol.M} we have $\CNF(\calM(\chi)) = \CNF(\top\land \neg A)
 =^\text{Obs.~\ref{obs.CNF}/(\ref{obs.CNF.andunit})} \CNF(\neg A)
 = \CNF(\calM(\chi^*))$.

If $\chi$ is of the form L/R and R/R as in
\[
\infer[\mathrm{cut}]{\Mseq{\Gamma_1, A}{\Gamma_2}{\Delta_1}{\Delta_2, A}}{
  \infer[\mathrm{wax}]{\Mseq{\Gamma_1, A}{\Gamma_2}{\Delta_1}{\Delta_2, A, A}}{}
  &
  \infer[\mathrm{wax}]{\Mseq{\Gamma_1, A}{\Gamma_2, A}{\Delta_1}{\Delta_2, A}}{}
}
\]
we reduce $\chi$ to $\chi^* =$
\[
\infer[\mathrm{wax}]{\Mseq{\Gamma_1, A}{\Gamma_2}{\Delta_1}{\Delta_2, A}}{}
\]
\ref{lem.ce_ipol.tamewred} is shown as above.
\ref{lem.ce_ipol.weight} is trivial.
For \ref{lem.ce_ipol.M} we have
$\CNF(\calM(\chi)) = \CNF(A\land \top)
=^\text{Obs.~\ref{obs.CNF}/(\ref{obs.CNF.andunit})} \CNF(A)
= \CNF(\calM(\chi^*))$.

The case of $\chi$ being of the form L/R and R/L as in
\[
\infer[\mathrm{cut}]{\Mseq{\Gamma_1, A}{\Gamma_2}{\Delta_1, A}{\Delta_2}}{
  \infer[\mathrm{wax}]{\Mseq{\Gamma_1, A}{\Gamma_2}{\Delta_1, A}{\Delta_2, A}}{}
  &
  \infer[\mathrm{wax}]{\Mseq{\Gamma_1, A}{\Gamma_2, A}{\Delta_1, A}{\Delta_2}}{}
}
\]
is ruled out by $\chi$ being tame.

The remaining cases of this cut-elimination argument can be found in
 the appendix.
\end{proof}

\begin{theorem}\label{thm.compl_prun_subs}
Let $\calC$ be a pruned interpolant of an implication $A\impl B$.
Then there is an {\LKminus} proof $\pi$ of $\Mseq{A}{}{}{B}$ with
 $\calC \subs \CNF(\calM(\pi))$.
\end{theorem}
\begin{proof}
By Lemma~\ref{lem.printerpolant_LKatproof} there is a tame {\LKat}
 proof $\pi$ of $\Mseq{A}{}{}{B}$ all of whose cuts are of type R with
 $\CNF(\calM(\pi)) = \calC$.
Then, by applying Lemma~\ref{lem.ce_ipol}, we obtain an
 {\LKminus} proof $\pi'$ with $\calC = \CNF(\calM(\pi)) \subs \CNF(\calM(\pi'))$.
\end{proof}

So even though interpolation in {\LKminus} is not complete as shown
 in Proposition~\ref{prop.Maehara_cutfree_incomplete}, we can still recover a desired interpolant $I$ in a
 restricted sense: after transforming $I$ into a pruned interpolant $\calC$
 we can obtain a proof whose interpolant is subsumed by $\calC$.
\begin{example}
The formula $p\land q \impl p\lor q$ has the four interpolants
 $p\land q, p, q, p \lor q$.
We know from the proof of Proposition~\ref{prop.Maehara_cutfree_incomplete} that
 the only interpolants obtainable from {\LKminus} proofs are $p$ and $q$.
The clause set $\{\{ p,q \}\}$, representing the formula $p\lor q$, is not a
 pruned interpolant.
The clause set $\{ \{ p \}, \{ q \} \}$, representing the formula $p\land q$,
 subsumes both $\{\{p\}\}$ and $\{\{ q \}\}$.
\end{example}

\begin{question}
Is standard interpolation in resolution complete up to subsumption for
 pruned interpolants?
\end{question}

\begin{question}
Can we extend these results to the sequent calculus $\mathbf{LJ}$ for the intuitionistic logic? How about other super intuitionistic or substructural logics?
\end{question}

\section{Sequent calculus with atomic cuts}\label{sec.LKat}

If we move from {\LKminus} to the slightly stronger {\LKat} we can even
 obtain a full completeness result.
This is achieved by a variant of the construction used in the proof
 of Lemma~\ref{lem.printerpolant_LKatproof}.

\begin{theorem} \label{Thm: LK with atomic cuts}
Maehara interpolation in $\LKat$ is complete.
\end{theorem}

\begin{proof}
Let $I$ be an interpolant of an implication $A\impl B$ and let the clause set
 $\calC = \{ C_1, \ldots, C_n \}$ be logically equivalent to $I$.
For $i=1,\ldots,n$ let $C_i = \{ \ell_{i,1},\ldots,\ell_{i,k_i} \}$.
We start by constructing proofs $\pi_i : \Mseq{A}{}{}{\ell_{i,1},\ldots,\ell_{i,k_i}}$ such that $\calM(\pi_i)$ is logically equivalent to $C_i$.
As $\calC$ is an interpolant of $A \to B$, we have
 $\LK \vdash A \Rightarrow \bigwedge_{i=1}^n C_i$.
Thus, for each $1 \leq i \leq n$, $\LK \vdash A \Rightarrow C_i$ and hence $\LK \vdash A \Rightarrow \ell_{i1}, \dots , \ell_{ik_i}$.
Let $\alpha_i$ be a cut-free proof of 
\[
\alpha_i : \qquad A ; \Rightarrow \ell_{i1}, \dots , \ell_{ik_i} ;
\]
By Observation \ref{observation}, $\calM (\alpha_i)$ is logically equivalent to $\bot$. Take the following proof tree as $\pi_i$, which moves all the literals in the succedent of $\alpha_i$ to the right-hand side of the Maehara partition only by cuts on literals and weakening rules:

\begin{center}
\AxiomC{$\alpha_i$}
\noLine
\UnaryInfC{$A ; \Rightarrow \ell_{i1}, \ell_{i2}, \dots , \ell_{ik_i} ;$}
\UnaryInfC{$A ; \Rightarrow \ell_{i1}, \ell_{i2}, \dots , \ell_{ik_i} ; \ell_{i1}$}
\AxiomC{$\ell_{i1} ; \Rightarrow ; \ell_{i1}$}
\doubleLine
\UnaryInfC{$A , \ell_{i1}; \Rightarrow \ell_{i2}, \dots , \ell_{ik_i} ; \ell_{i1}$}
\BinaryInfC{$A ; \Rightarrow \ell_{i2}, \dots , \ell_{ik_i} ; \ell_{i1}$}
\hspace{-0.5em}\DisplayProof

\AxiomC{$\vdots$}
\UnaryInfC{$A ; \Rightarrow \ell_{ik_i} ; \ell_{i1}, \dots , \ell_{ik_i-1}$}
\UnaryInfC{$A ; \Rightarrow \ell_{ik_i} ; \ell_{i1}, \dots , \ell_{ik_i-1}, \ell_{ik_i}$}
\AxiomC{$\ell_{ik_i} ; \Rightarrow ; \ell_{ik_i}$}
\doubleLine
\UnaryInfC{$A , \ell_{ik_i}; \Rightarrow  ; \ell_{i1}, \dots \ell_{ik_i}$}
\BinaryInfC{$A ; \Rightarrow  ; \ell_{i1}, \dots \ell_{ik_i}$}
\hspace{4em}\DisplayProof
\end{center}

\normalsize \noindent The double lines in the proof mean using the rules $(Lw)$ and $(Rw)$ as often as needed. To construct the proof $\pi_i$, we start with $\alpha_i$ and use the rule $(Rw)$ on its end-sequent. Then, we take the valid sequent $\ell_{i1}; \Rightarrow ; \ell_{i1}$ and apply $(Lw)$ and $(Rw)$ as needed. Then we can use the cut rule to move $\ell_{i1}$ to the right-hand side of the Maehara partition. We repeat this procedure for the rest of the literals to get $\pi_i$. 
To see that $\calM(\pi_i)$ is logically equivalent to $C_i$, note that each cut in the proof $\pi_i$ is on a literal and the cut formula is always on the left-hand side of the Maehara partition. Moreover, $V(I_i) \subseteq V(I) \subseteq V(A) \cap V(B) \subseteq V(A)$. Hence, in each cut, the cut formula is in $V(A)$, and by the Maehara interpolation algorithm for $\LKlit$, the interpolant of the conclusion of the cut rule is the disjunction of the interpolants of the premises. Since $\calM(\alpha_1)$ is logically equivalent to $\bot$ and $\calM(\ell_{ij} ; \Rightarrow ; \ell_{ij})= \ell_{ij}$ for $1 \leq j \leq k_i$, we get $\calM(\pi_i)$ is logically equivalent to $ \bot \vee \ell_{i1} \vee \dots \vee \ell_{ik_i}$, which is logically equivalent to $\ell_{i1} \vee \dots \vee \ell_{ik_i}= \bigvee_{j=1}^{k_i} \ell_{ij}= C_i$.
Now apply Lemmas~\ref{lem.LKatLKlitEquivInt} and~\ref{lem.interpolant_conjunction_construction} 
to obtain  an {\LKat} proof $\pi$ of $\Mseq{A}{}{}{B}$ with $\calM(\pi)$ logically equivalent to  $\calC$.
\end{proof}
A trivial consequence of Theorem \ref{Thm: LK with atomic cuts} is that Maehara interpolation in $\LKlit$ and $\LKm$ is also complete. 
\begin{remark}
It is worth noting that the proof of Theorem \ref{Thm: LK with atomic cuts} provides an interpolant logically equivalent to the given interpolant $I$ in
 CNF, only by adding $\top$ as conjuncts and $\bot$ as disjuncts.
Thus, we have even shown that Maehara interpolation is syntactically complete
 up to unit elimination of conjunction and disjunction for formulas in CNF.
\end{remark}

\begin{example}
We find a proof $\pi: p \wedge q; \Rightarrow ;p \vee q$ in $\LKat$ such that $\calM(\pi) = p \wedge q$. Denote $I_1= p$ and $I_2= q$. We have 
\begin{center}
\begin{tabular}{c c c c }
$\pi_1:$ \;
\AxiomC{$p ; \Rightarrow ; p$}
\UnaryInfC{$p \wedge q ; \Rightarrow ; p$}
\DisplayProof
&  &
$\pi_2:$ \;
\AxiomC{$q ; \Rightarrow ; q$}
\UnaryInfC{$p \wedge q ; \Rightarrow ; q$}
\DisplayProof
\end{tabular}
\end{center}
and $\calM(\pi_1)=p$ and $\calM(\pi_2)=q$. Take the following proof tree $\sigma: p \wedge q; \Rightarrow ;p \vee q$ in $\LKat$ where all the cuts are context-splitting:
\begin{prooftree}
\AxiomC{$\pi_2$}
\noLine
\UnaryInfC{$p \wedge q ; \overset{q}{\Longrightarrow} ; q$}
\AxiomC{$\pi_1$}
\noLine
\UnaryInfC{$p \wedge q ; \overset{p}{\Longrightarrow} ; p$}
\AxiomC{$; p \overset{\top}{\Longrightarrow} ; p$}
\UnaryInfC{$; p, q \overset{\top}{\Longrightarrow} ; p$}
\UnaryInfC{$; p, q \overset{\top}{\Longrightarrow} ; p\vee q$}
\doubleLine
\BinaryInfC{$p \wedge q ; q \overset{p \wedge \top}{\Longrightarrow} ; p \vee q$}
\doubleLine
\BinaryInfC{$p \wedge q, p \wedge q ;  \overset{q \wedge p \wedge \top}{\Longrightarrow} ; p \vee q$}
\UnaryInfC{$ p \wedge q ;  \overset{q \wedge p \wedge \top}{\Longrightarrow} ; p \vee q$}
\end{prooftree}
where the double lines mean applying the left and right weakening rules (to get the same contexts in the premises of the cut rule) and then the cut rule. We have $\calM(\sigma)$ is logically equivalent to $p \wedge q$. 

\end{example}


\section{Propositional normal modal logics}\label{sec.modal}
We work with the language $\mathcal{L}_\Box= \{\bot, \wedge, \vee, \neg, \Box\}$. The modal rules we consider are:
\begin{center}
\begin{tabular}{c c c c}
\AxiomC{$\Gamma \Rightarrow A$}
  \RightLabel{\scriptsize{($K$)}}
 \UnaryInfC{$\Box \Gamma \Rightarrow \Box A$}
 \DisplayProof 
 &
  \AxiomC{$\Gamma \Rightarrow $}
  \RightLabel{\scriptsize{($D$)}}
 \UnaryInfC{$\Box \Gamma \Rightarrow $}
 \DisplayProof 
 &
 \AxiomC{$\Gamma, \Box \Gamma \Rightarrow A$} 
  \RightLabel{\scriptsize{($4$)}}
 \UnaryInfC{$\Box \Gamma \Rightarrow \Box A$}
 \DisplayProof
 &
 \AxiomC{$\Gamma , A \Rightarrow \Delta$}
\RightLabel{\scriptsize{($T$)}}
 \UnaryInfC{$\Gamma , \Box A \Rightarrow \Delta$}
 \DisplayProof
 \end{tabular}
 \end{center}

Consider the normal modal logics $\mathsf{K}$, $\mathsf{D}$, $\mathsf{T}$, $\mathsf{K4}$, $\mathsf{KD4}$, and $\mathsf{S4}$. Take the usual sequent calculi for these logics by adding the corresponding modal rules to {\LK}:
\begin{center}
$\mathbf{K}= \LK+(K) \quad \mathbf{KD}= \LK+(K)+(D) \quad \mathbf{KT}= \LK+(K)+(T)$\\
$\quad \mathbf{K4}= \LK+(4) \quad \mathbf{KD4}= \LK+(4)+(D) \quad \mathbf{S4}= \LK+(T)+(4)$
\end{center}

We extend the interpolation algorithm with the following rules:

\begin{center}
\begin{tabular}{c c}
 $\Box A;  \overset{\bot}{\Longrightarrow}   \Box A; $ & $;  \Box A \overset{\top}{\Longrightarrow} ;  \Box A$  \\
 $\Box A; \overset{ \Box A}{\Longrightarrow} ;   \Box A$    & $;  \Box A \overset{\neg  \Box A}{\Longrightarrow}  \Box A; $ 
\end{tabular}
\end{center}

If the last rule is $(K)$:
\begin{center}
\begin{tabular}{c c}
 \AxiomC{$\Gamma_1;  \Gamma_2 \overset{C}{\Longrightarrow}  ; A$}
 \UnaryInfC{$\Box \Gamma_1; \Box \Gamma_2 \overset{\Box C}{\Longrightarrow}  ; \Box A$}
 \DisplayProof
     &  
\AxiomC{$\Gamma_1;  \Gamma_2 \overset{C}{\Longrightarrow}   A ;$}
 \UnaryInfC{$\Box \Gamma_1; \Box \Gamma_2 \overset{\neg \Box  \neg C}{\Longrightarrow}   \Box A ;$}
 \DisplayProof
\end{tabular}
\end{center}
If the last rule is $(D)$, then note that the rule $(K)$ is also present in the calculus. Then:
\begin{center}
 \AxiomC{$\Gamma_1;  \Gamma_2 \overset{C}{\Longrightarrow}  $}
 \UnaryInfC{$\Box \Gamma_1; \Box \Gamma_2 \overset{\Box C}{\Longrightarrow} $}
 \DisplayProof
\end{center}
If the last rule is $(4)$:
\begin{center}
\begin{tabular}{c c}
 \AxiomC{$\Gamma_1, \Box \Gamma_1;  \Gamma_2, \Box \Gamma_2 \overset{C}{\Longrightarrow}  ; A$}
 \UnaryInfC{$\Box \Gamma_1; \Box \Gamma_2 \overset{\Box C}{\Longrightarrow}  ; \Box A$}
 \DisplayProof
     &  
\AxiomC{$\Gamma_1, \Box \Gamma_1;  \Gamma_2, \Box \Gamma_2 \overset{C}{\Longrightarrow}   A ;$}
 \UnaryInfC{$\Box \Gamma_1; \Box \Gamma_2 \overset{\neg \Box  \neg C}{\Longrightarrow}   \Box A ;$}
 \DisplayProof
\end{tabular}
\end{center}
If the last rule is $(T)$:
\begin{center}
\begin{tabular}{c c}
 \AxiomC{$\Gamma_1, A;  \Gamma_2 \overset{C}{\Longrightarrow}  \Delta_1 ; \Delta_2$}
 \UnaryInfC{$\Gamma_1, \Box A;  \Gamma_2 \overset{C}{\Longrightarrow}  \Delta_1 ; \Delta_2$}
 \DisplayProof
     &  
\AxiomC{$\Gamma_1;  A, \Gamma_2 \overset{C}{\Longrightarrow}  \Delta_1 ; \Delta_2$}
 \UnaryInfC{$\Gamma_1;  \Box A, \Gamma_2 \overset{C}{\Longrightarrow}  \Delta_1 ; \Delta_2$}
 \DisplayProof
\end{tabular}
\end{center}

Define \emph{modal literals}, denoted by $\ell_\Box$, as  $p, \neg p, \Box A, \neg \Box A$, where $p$ is a propositional atom and $A$ is a modal formula. The \emph{conjunctive normal form} of a modal formula, denoted by $\mCNF$,  is defined similarly to the propositional case: conjunctions of disjunctions of modal literals. A modal formula of the form $\Box B$ is an \emph{immediate modal subformula} of a modal formula $A$ if it is a subformula of $A$ and there is no other subformula of $A$ that contains $\Box B$.
By an easy observation, we see that every modal formula has a logically equivalent modal conjunctive normal form. The reason is that we can see immediate modal subformulas of a formula as new atomic formulas (i.e., not occurring in $A$) and then this formula will be propositional and it has a $\CNF$. Transforming the new atoms back to the immediate modal subformulas will provide an $\mCNF$ for $A$. 

\begin{proposition}\label{modal.prop.Maehara_cutfree_incomplete}
Maehara interpolation in cut-free propositional $\mathbf{K}$ is not complete. 
\end{proposition}

\begin{proof}
The formula $\Box(p \wedge q) \to \Box(p \vee q)$, which is valid in propositional $\mathbf{K}$, has the interpolants $\Box(p \wedge q)$, $\Box(p \vee q)$, and $\Box p \wedge \Box q$ but neither of them can be read off of a cut-free proof.     
\end{proof}

\begin{theorem} \label{Thm: K with atomic cuts}
Let $G\in \{\mathbf{K}, \mathbf{KD}, \mathbf{KT}, \mathbf{K4}, \mathbf{KD4}, \mathbf{S4}\}$. Maehara interpolation $\calM$ in $G$ with cuts on atomic formulas and boxed formulas is complete.
\end{theorem} 

\begin{proof}
Similar to the proof of Theorem \ref{Thm: LK with atomic cuts}. Suppose $I$ is an interpolant of the $G$-provable sequent $A \Rightarrow B$. Take the $\mCNF$ of interpolant $I=\bigwedge_{i=1}^n I_i= \bigwedge_{i=1}^n \bigvee_{j=1}^{k_i} \ell_{ij}$, where each $\ell_{ij}$ is a modal literal. We can construct the proofs $\pi_i$  as in Theorem \ref{Thm: LK with atomic cuts} and every step of the proof is similar to before. Note that the cut rule is always monochromatic.
\end{proof}
It is interesting to investigate these questions for non-normal modal logics, and other logics such as the G\"{o}del-l\"{o}b logic $\mathsf{GL}$.

\section{First-order logic}\label{sec.fol}

We now move to completeness of interpolation algorithms in first-order logic.
There are different strategies for interpolation in first-order 
 logic (with and without equality).
The most prominent strategy consists in  1.\ computing a propositional 
 interpolant and 2.\ introducing quantifiers into this propositional 
 interpolant in order to convert it to a first-order interpolant.
Such interpolation algorithms can be found, e.g.,
 in~\cite[Theorem 13]{Kreisel67Elements},
 \cite{Huang95Constructing},
 \cite[Section 5.13]{Harrison09Handbook}, and
 \cite[Section 8.2]{Baaz11Methods}.
An alternative strategy consists in 1.\ replacing function symbols by 
 relation symbols in the input formula, 2.\ applying an (almost) 
 propositional interpolation algorithm, and 3.\ translating the interpolant
 thus obtained back into a language with function symbols.
Such algorithms can be found, e.g., in \cite{Craig57Three} and
 \cite[Section 7.3]{Avigad23Mathematical}.

What kind of results can we expect for these two strategies? 
We can clearly obtain incompleteness results for the cut-free sequent
 calculus for first-order logic as a straightforward extension of
 Proposition~\ref{prop.Maehara_cutfree_incomplete}.
However, when we consider sequent calculus with atomic cuts, the answer
 is not clear at first sight and boils down to the question of whether
 there is some source of incompleteness on the first-order level.

In this section, we will prove a strong incompleteness result for the first
 strategy: we will show that it is incomplete, regardless
 of the concrete algorithms employed in the two phases and the concrete
 calculus used for obtaining the interpolant in the first phase.
In order to do this we first have to make this strategy precise.
To that aim, we consider the language $\mathcal{L}= \{\bot, \wedge, \vee, \neg, \exists, \forall\}$ and write $L_P(A)$ for the set of all predicate
 symbols occurring in the first-order formula $A$.
If we add the following rules to propositional $\LK$, we get first-order $\LK$:
\begin{center}
\begin{tabular}{c c}
\vspace{5pt}
\AxiomC{$A[x/t], \Gamma \Rightarrow \Delta$}
 \RightLabel{\small{$(L \forall)$} }
\UnaryInfC{$\forall x A, \Gamma \Rightarrow \Delta$}
\DisplayProof
&  \AxiomC{$\Gamma \Rightarrow \Delta, A[x/y]$}
 \RightLabel{\small{$(R \forall)$} }
\UnaryInfC{$\Gamma \Rightarrow \Delta, \forall x A$}
 \DisplayProof \\
\AxiomC{$A[x/y], \Gamma \Rightarrow \Delta$}
 \RightLabel{\small{$(L \exists)$} }
\UnaryInfC{$\exists x A, \Gamma \Rightarrow \Delta$}
 \DisplayProof     & 
 \AxiomC{$\Gamma \Rightarrow \Delta, A[x/t]$}
 \RightLabel{\small{$(R \exists)$} }
\UnaryInfC{$\Gamma \Rightarrow \Delta, \exists x A$}
 \DisplayProof  
\end{tabular}
\end{center}
where $t$ is an arbitrary term and in $(R \forall)$ and $(L \exists)$, the variable $y$ is not free in the conclusion.
\begin{definition}
Let $A\impl B$ be a first-order formula.
A formula $C$ is called {\em weak interpolant} of $A \impl C$ if
 $\models A\impl C$, $\models C\impl B$, and $L_P(C)\subseteq L_P(A)\cap L_P(B)$.
\end{definition}
So while a weak interpolant satisfies the usual language condition on the
 predicate symbols, it may contain constant and function symbols which
 are not in the intersection of the languages of $A$ and $B$.
\begin{example}\label{ex.weak_ipol}
Define the formulas
\begin{align*}
A & = \forall v\, (v < f(v)) \land \forall v \forall w\, (Z(v) \land v < w \impl \neg Z(w)) \\
B & = Z(c) \impl \exists u \exists v\, (Z(u) \land \neg Z(v))
\end{align*}
Then $L(A) = \{ Z,<,f \}$, $L(B) = \{ Z,c \}$, thus $L(A)\cap L(B) = \{ Z \}$.
Then $C_0 = Z(c) \impl Z(f(c))$ is a weak interpolant.
\end{example}
An algorithm for the computation for interpolants in propositional logic
 can usually be easily adapted to compute weak interpolants in first-order
 logic.
For example, we can add the rules
\begin{center}\begin{tabular}{c c}
\vspace{8pt}

\AxiomC{$\Mseq[C]{A[x/y],\Gamma_1}{\Gamma_2}{\Delta_1}{\Delta_2}$}
\UnaryInfC{$\Mseq[C]{\exists x A(x), \Gamma_1}{\Gamma_2}{\Delta_1}{\Delta_2}$}
\DisplayProof     & 
\AxiomC{$\Mseq[C]{\Gamma_1}{A[x/y],\Gamma_2}{\Delta_1}{\Delta_2}$}
\UnaryInfC{$\Mseq[C]{\Gamma_1}{\exists x A(x), \Gamma_2}{\Delta_1}{\Delta_2}$}
\DisplayProof \\

\AxiomC{$\Mseq[C]{\Gamma_1}{\Gamma_2}{\Delta_1, A[x/t]}{\Delta_2}$}
\UnaryInfC{$\Mseq[C]{\Gamma_1}{\Gamma_2}{\Delta_1, \exists x A(x)}{\Delta_2}$}
\DisplayProof &
\AxiomC{$\Mseq[C]{\Gamma_1}{\Gamma_2}{\Delta_1}{\Delta_2, A[x/t]}$}
\UnaryInfC{$\Mseq[C]{\Gamma_1}{\Gamma_2}{\Delta_1}{\Delta_2, \exists x A(x)}$}
\DisplayProof  
\end{tabular}
\end{center}

and analogous rules for $\forall$ to the algorithm from Section~\ref{ssec.seq_calc}
 in order to obtain an algorithm
 that computes a weak interpolant of $A\impl B$ from an {\LKm} proof
 of $\Mseq{A}{}{}{B}$.
We can now make the first phase of the strategy precise: 1.\ we compute
 a weak interpolant $C_0$ for $A\impl B$, e.g., as in the algorithm described above.
For describing the second phase we define:
\begin{definition}
We define the binary relation {\em $A$ is an abstraction of $B$} on
 first-order formulas as the smallest reflexive and transitive relation
 that satisfies the following condition:
If $A = Qx\, A_0$ for a $Q \in \{ \forall,\exists \}$ and $B = A_0\unsubst{x}{t}$ for some term $t$ then $A$ is an abstraction of $B$.
\end{definition}
\begin{lemma}
There is an algorithm $\calB$ which, given a weak interpolant $C_0$ of $A\impl B$ computes an interpolant $C$ of $A\impl B$ which is an abstraction
 of $C_0$.
\end{lemma}
\begin{proof}[Proof Sketch]
$\calB$ replaces a maximal term $t$ which is not in $L(A)\cap L(B)$ by a
 new bound variable $x$ which is either existentially or universally 
 quantified, depending on whether the leading function symbol of $t$
 is in $L(A) \setminus L(B)$ or $L(B)\setminus L(A)$.
If $C[t]$ is a weak interpolant of $A\impl B$, then $Q x\, C[x]$ is an 
 abstraction of $C[t]$ and a weak interpolant of $A\impl B$ which
 contains one less term violating the language condition.
By repeating this step one obtains an interpolant in the usual sense.
See, e.g.,~\cite[Lemma 8.2.2]{Baaz11Methods} for a detailed exposition of this
 proof.
\end{proof}
We can now make the second phase precise: 2.\ we apply the algorithm $\calB$
 to the weak interpolant $C_0$ of $A\impl B$ in order to obtain an
 interpolant $C := \calB(C_0)$ of $A\impl B$.
\begin{example}
Continuing Example~\ref{ex.weak_ipol}
 we obtain the interpolant $\calB(C_0) = \forall x\exists y\, (Z(x) \impl \neg Z(y))$ of $A\impl B$.
\end{example}
The central observation is now the following: in first-order logic,
 there are interpolants which are not abstractions of weak interpolants.
More precisely:
\begin{lemma}\label{lem.ipol_not_abstraction_of_wipol}
There are first-order sentences $A$, $B$, and $C$ s.t.\ $C$ is an
 interpolant of $A\impl B$ but $C$ is not an abstraction of a
 quantifier-free weak interpolant of $A\impl B$.
\end{lemma}
\begin{proof}
Let $A := \forall x\, (I(0)\land (I(x)\impl I(s(x))))$, let 
 $B:=I(s(s(0)))$, and let $C := A$.
Then $L(A) = L(B) = L(C) = \{ 0,s,I \}$, $A\impl B$ is a valid formula,
 and $C$ is an interpolant of $A\impl B$.

Suppose that $C$ is abstraction of a quantifier-free weak interpolant $C_0$.
Then $C_0$ is of the form $I(0) \land (I(t)\impl I(s(t)))$ for some term $t$
 and we would have $\models C_0 \impl B$.
However,
\[
I(0) \land (I(t) \impl I(s(t))) \impl I(s(s(0)))
\]
is not a valid formula which can be shown easily by a countermodel 
 $\mathcal{N}$ with domain $\mathbb{N}$ s.t.\ $0 \in I^\mathcal{N}$,
 $2 \notin I^\mathcal{N}$, and $1\in I^\mathcal{N}$ iff $t^\mathcal{N}=0$.
\end{proof}
Therefore, any algorithm that computes only abstractions of weak interpolants
 is incomplete.
In particular: let $\calM$ be the interpolation algorithm for first-order {\LKminus} from Section 8.2 in~\cite{Baaz11Methods}.
Clearly $\calM$ is not complete due to Proposition~\ref{prop.Maehara_cutfree_incomplete}.
Let $\calM'$ be the (straightforward) extension of $\calM$ to first-order {\LKat}.
 Then we obtain:
\begin{theorem}
$\calM'$ is not complete.
\end{theorem}
\begin{proof}
Let $A\impl B$ and $C$ be as in Lemma~\ref{lem.ipol_not_abstraction_of_wipol}.
Then $\calM(\pi')$ is an abstraction of a weak interpolant of $A\impl B$
 and hence different from $C$.
\end{proof}

\begin{question}
Are interpolation algorithms following the second strategy incomplete?
\end{question}

\section{A remark on Beth's definability theorem}\label{sec.Beth}

Beth's definability theorem is one of the most important applications of
 interpolation in mathematical logic.
As we will briefly point out in this section, completeness properties of 
 the interpolation theorem apply directly to Beth's definability theorem.
For the results in this section it will be convenient to explicitely indicate
 all predicate symbols that occur in a first-order formula by writing
 $A(R_1,\ldots,R_n)$.

\begin{definition}
Let $R, R_1, \ldots, R_n$ be predicate symbols.
A sentence $A(R,R_1,\ldots,R_n)$ is an {\em implicit definition} of $R$ if
\begin{equation}\label{eq.imp_def}
A(R,R_1,\ldots,R_n) \land A(R',R_1,\ldots,R_n) \impl \forall \vec{x} (R(\vec{x}) \liff R'(\vec{x}))
\end{equation}
is valid.

$A(R,R_1,\ldots,R_n)$ is an {\em explicit definition} of $R$ if there is a formula $F(\vec{x})$ s.t.
\begin{equation}\label{eq.exp_def}
A(R,R_1,\ldots,R_n) \impl \forall \vec{x} ( R(\vec{x}) \liff F(\vec{x}))
\end{equation}
is valid.
\end{definition}
Beth's definability theorem states that whenever $R$ is definable implicitly (in first-order logic), then $R$ is also definable explicitly.
We first observe that (\ref{eq.imp_def}) is valid iff
\begin{equation}\label{eq.Beth}
(A(R,R_1,\ldots,R_n) \land R(\myvec{x})) \impl (A(R',R_1,\ldots,R_n) \impl R'(\myvec{x}))    
\end{equation}
is valid.
Then, in analogy to Definition~\ref{def.ipol_algo_complete}, we can say that
 an algorithm
 $\calD$ which receives a proof $\pi$ of (\ref{eq.Beth}) as input and returns
 an $F$ s.t.\ (\ref{eq.exp_def}) is valid is complete if for every
 $F$ there is a $\pi$ with $F = \calD(\pi)$.

The standard proof of Beth's definability theorem
 from the interpolation theorem, see, e.g.~\cite{Takeuti87Proof}, 
 now proceeds as follows:
Let $A(R,R_1,\ldots,R_n)$ be an implicit definition of $R$ and let
 $\pi$ be a proof of (\ref{eq.Beth}).
Then applying an interpolation algorithm $\calI$ to $\pi$ yields a
 formula $F = \calI(\pi)$ with $L_P(F)\subseteq \{ R_1,\ldots,R_n\}$ such that both
\[
A(R,R_1,\ldots,R_n) \land R(\myvec{x}) \impl F(\myvec{x})
\]
and, by renaming $R'$ to $R$,
\[
F(\myvec{x})\impl ( A(R,R_1,\ldots,R_n) \impl R(\myvec{x}) )
\]
are valid. Hence also
\[
A(R,R_1,\ldots,R_n) \impl \forall \vec{x} (R(\vec{x}) \liff F(\vec{x}))
\]
is valid, so $F(\vec{x})$ is an explicit definition of $R$.
Writing $\calD_\calI$ for the algorithm that takes a proof
 (\ref{eq.Beth}) and returns an explicit definition $F(\vec{x})$
 of $R$ we see that $\calD_\calI$ is the restriction of $\calI$ to formulas
 of the form (\ref{eq.Beth}).
In particular:
\begin{observation}
$\calD_\calI$ is complete iff $\calI$ is complete
 on formulas of the form (\ref{eq.Beth}).
\end{observation}

\section{Conclusion}

We have initiated the study of completeness properties of interpolation
 algorithms by proving several results about some of the most important
 interpolation algorithms:
The standard algorithms for resolution and cut-free sequent calculus for
 propositional logic are incomplete.
On the other hand, in the sequent calculus with atomic cuts, it is complete.
Moreover, even in the cut-free sequent calculus one can obtain a weaker
 completeness result: completeness of pruned interpolants up to subsumption.
We have also extended our results to normal modal logics and to first-order
 logic and found a new source of incompleteness in first-order logic
 that applies to a wide variety of interpolation algorithms.
We have also shown that completeness properties of interpolation algorithms
 correspond directly to completeness properties of Beth's definability
 theorem.

These results show that the completeness of
 an interpolation algorithm is related to the amount of freedom, or redundancy,
 that is permitted by a proof system and its subtle interplay with
 the interpolation algorithm, as witnessed very clearly, e.g., by the proof of 
 Theorem~\ref{Thm: LK with atomic cuts}, the completeness of interpolation in
 the sequent calculus with atomic cut.

Moreover, our results show clearly that the completeness of a proof
 calculus (w.r.t.\ some semantics) is a different question from
 that of the completeness of an interpolation algorithm in this calculus.
For example, both cut-free sequent calculus and sequent calculus with atomic cuts are complete w.r.t.\ Tarski semantics.
However, the standard interpolation algorithm is complete in the latter but
 not in the former.
 
We believe that this work is merely a first step in a wider project
 of gauging the expressive power of interpolation algorithms based
 on their completeness properties.
We have already mentioned many open questions in the paper.
We consider the following open problems to be the most promising and relevant:
We plan to investigate the completeness of interpolation algorithms for local 
 proofs~\cite{Jhala06Practical,Kovacs09Interpolation} which are of particular
 relevance in the CAV community.
In order to get a better picture of the situation in first-order logic,
 it would be useful to investigate the second interpolation strategy
 for first-order proofs as mentioned in Section~\ref{sec.fol}.
Moreover, we are intrigued by the question whether, or respectively: in how far,
 the cut-elimination argument underlying Theorem~\ref{thm.compl_prun_subs}
 can be extended to first-order logic.
It would be interesting to investigate these questions also for
 intuitionistic logic where, due to the more restricted availability
 of normal forms, quite different techniques will presumably be needed.
Resolution with weakening for classical propositional logic is interesting from a proof-theoretic
 point of view since, in contrast to ordinary resolution, it is complete (as a proof calculus, w.r.t.\ standard semantics).
However, the completeness of its interpolation algorithm is unknown.
Also, it would be interesting to relate these completeness and incompleteness
 results more directly to applications in verification.
We leave these, and the questions mentioned throughout the paper, to future work.

\bibliographystyle{plain}
\bibliography{references}

\newpage

\section{Appendix}

In this appendix we collect proofs that had to be omitted
 from the main text because of space limitations.
For the convenience of the reader we also restate the lemmas and
 theorems proved here.

\ThmSoundnessMaehara*
\begin{proof}
As $\pi$ is a proof of the split sequent $\Gamma_1;\Gamma_2 \Rightarrow \Delta_1;\Delta_2$ in  $G$, by Remark \ref{Rem: monochromatic atomic cut} all the cuts in $\pi$ are monochromatic. We prove that the Maehara algorithm outputs a formula $\calM(\pi)=C$, called the \emph{interpolant}, such that $V(C) \subseteq V(\Gamma_1 \cup \Delta_1) \cap V(\Gamma_2 \cup \Delta_2)$ and $G \vdash \Gamma_1 \Rightarrow \Delta_1, C$ and $G\vdash C , \Gamma_2 \Rightarrow \Delta_2$. Moreover, $|\calM(\pi)| \leq |\pi|$. For the axioms and each rule except the cut rule, see \cite[Section 4.4.2]{Troelstra}. As for the right rules, let us only investigate $(R \neg)$: \begin{center} \begin{tabular}{c c} \AxiomC{$\Gamma_1, A ; \Gamma_2 \overset{C}{\Longrightarrow} \Delta_1; \Delta_2$} \UnaryInfC{$\Gamma_1; \Gamma_2 \overset{C}{\Longrightarrow} \Delta_1, \neg A ; \Delta_2$} \DisplayProof     &   \AxiomC{$\Gamma_1; \Gamma_2, A \overset{C}{\Longrightarrow} \Delta_1 ; \Delta_2$} \UnaryInfC{$\Gamma_1; \Gamma_2 \overset{C}{\Longrightarrow} \Delta_1 ; \Delta_2, \neg A$} \DisplayProof \end{tabular} \end{center}
The reason is that for the left case by the induction hypothesis, we have 
\[
G \vdash \Gamma_1, A \Rightarrow C, \Delta_1 \quad G \vdash \Gamma_2, C \Rightarrow \Delta_2
\]
Applying the rule $(R \neg)$ we get 
\[
G \vdash \Gamma_1 \Rightarrow C, \neg A,  \Delta_1 \quad G \vdash \Gamma_2, C \Rightarrow \Delta_2
\]
which means that $C$ works as the interpolant of the conclusion, as desired. Similarly for the right case. \\
Now, we investigate the cut rule. Suppose the last rule in $\pi$ is an instance of an atomic cut, a cut on a literal, or a monochromatic (in particular analytic) cut, depending on whether $G$ is $\LKat$, $\LKlit$, or $\LKm$. As $A$ is the cut formula,  $V_1= V(\Gamma_1 \cup \Delta_1)$ and $V_2= V(\Gamma_2 \cup \Delta_2)$, by Remark \ref{Rem: monochromatic atomic cut}, for any atomic cut rule, or a cut on a literal, with the cut formula $A$, w.l.o.g. we can assume it is monochromatic, i.e., either $A \in V_1$ or $A \in V_2$.
In the former case when $A \in V_1$, we choose the case where $A$ appears on the left-hand side of the Maehara partitions in the premises. Otherwise, if $A \in V_2$ we take the case where $A$ appears on the right-hand side of the Maehara partitions in the premises. Let us investigate one such case. If $A \in V_1$, we define the interpolant of the conclusion of the cut rule as $C \vee D$ because $V(C) \subseteq (V_1 \cup \{A\}) \cap V_2= V_1 \cap V_2$ and $V(D) \subseteq (V_1 \cup \{A\}) \cap V_2= V_1 \cap V_2$. Hence, the variable condition for $C \vee D$ holds. Moreover, by the induction hypothesis, the following sequents are provable in $G$
\begin{center}
\begin{tabular}{c c c c}
$\Gamma_1 \Rightarrow \Delta_1, A, C$  & $C , \Gamma_2 \Rightarrow \Delta_2$  &  $\Gamma_1, A \Rightarrow \Delta_1, D$ & $D , \Gamma_2 \Rightarrow \Delta_2$
\end{tabular}
\end{center}
Hence, the following are provable in $G$
\begin{center}
\begin{tabular}{c c}
$ \Gamma_1 \Rightarrow \Delta_1, C \vee D$     &  $ C \vee D , \Gamma_2 \Rightarrow \Delta_2$
\end{tabular}
\end{center}
Similarly, if $A \in V_2$,  it is easy to see that $E \wedge F$ works as the interpolant. It is easy to see that $|\calM(\pi)| \leq |\pi|$.
\end{proof}

\ObsCNF*
\begin{proof}
For \ref{obs.CNF.logeq}, we use induction on the structure of $A$.
Let $A=\top$, then $\CNF(\top)=\{\}$ and the formula interpretation of
 $\{\}$ is $\top$.
Similarly, if $A=\bot$, then the formula interpretation of
 $\CNF(\bot)=\{\emptyset\}$ is $\bot$.
If $A$ is a literal, we have $\CNF(\ell)=\{\ell\}$, where its formula 
 interpretation is $\ell$.
Let $A = B \wedge C$.
We have $\CNF(B \wedge C)= \CNF(B) \cup \CNF(C)$ and by the induction hypothesis
 it is logically equivalent to $B \wedge C$.
Let $A = B \vee C$ and $\CNF(B)=\{B_1, \dots, B_n\}$ where each
 $B_i=\{\ell_{i1}, \dots, \ell_{ik_i}\}$ and $\CNF(C)=\{C_1, \dots, C_m\}$
 where each $C_r=\{\ell'_{r1}, \dots, \ell'_{ru_r}\}$.
The formula interpretation of $\CNF(B)$ is $\bigwedge_{i=1}^n \bigvee_{j=1}^{k_i} \ell_{ij}$
 and the formula interpretation of $\CNF(C)$ is $\bigwedge_{r=1}^m \bigvee_{s=1}^{u_r} \ell'_{rs}$.
Then,  $\CNF(B \vee C)=\CNF(\psi) \times \CNF(\chi) = \{C \cup D \mid C \in \CNF(\psi), D \in \CNF(\chi)\}$. The formula interpretation of $\CNF(\psi \vee \chi)$ is 
\begin{center}
$(\bigvee_j \ell_{ij} \vee \bigvee_s \ell'_{1s}) \wedge \dots \wedge (\bigvee_j \ell_{ij} \vee \bigvee_s \ell'_{ms}) \wedge \dots \wedge$ \\
$(\bigvee_j \ell_{nj} \vee \bigvee_s \ell'_{1s}) \wedge \dots \wedge (\bigvee_j \ell_{nj} \vee \bigvee_s \ell'_{ms})$
\end{center}
which is logically equivalent to 
\[
(\bigvee_j \ell_{1j} \wedge \bigvee_j \ell_{2j} \wedge \dots \wedge \bigvee_j \ell_{nj}) \vee (\bigvee_s \ell'_{1s} \wedge \bigvee_s \ell'_{2s} \wedge \dots \wedge \bigvee_s \ell'_{ms})
\]
which is $(\bigwedge_{i=1}^n \bigvee_{j=1}^{k_i} \ell_{ij}) \vee (\bigwedge_{r=1}^m \bigvee_{s=1}^{u_r} \ell'_{rs})$ that is logically equivalent to $B \vee C$ by the induction hypothesis.

\noindent For \ref{obs.CNF.andunit}, $\CNF(A) \wedge \top=\CNF(A) \cup \CNF(\top)=\CNF(A) \cup \emptyset=\CNF(A)$.
For \ref{obs.CNF.orunit}, $\CNF(A) \vee \bot=\CNF(A) \times \CNF(\bot)=\CNF(A) \times \{\emptyset\}=\{C \cup \emptyset \mid C \in \CNF(A)\}= \CNF(A)$.

The rest are easy and left to the reader.
\end{proof}

\LemNegInversion*
\begin{proof}
We prove $(1)$, the others are similar. We use induction on the depth of $\pi$. As the base case, suppose $\pi$ is an axiom. Then, it has one of the following forms concerning the Maehara partition:
\[
\neg A ; \overset{\neg A}{\Longrightarrow} ; \neg A \qquad ; \neg A \overset{\top} {\Longrightarrow} ; \neg A
\]
Then, for the left (right) case, we choose $\pi'$ as the following left (right) proof tree:
\begin{center}
\begin{tabular}{c c}
\AxiomC{$; A \overset{\neg A}{\Longrightarrow} A ;$}
\UnaryInfC{$\neg A; A \overset{\neg A}{\Longrightarrow} ;$}
\DisplayProof
&
\AxiomC{$; A \overset{\top}{\Longrightarrow}  ; A$}
\UnaryInfC{$; \neg A, A \overset{\top}{\Longrightarrow} ;$}
\DisplayProof
\end{tabular}
\end{center}
It is easy to see that the Maehara algorithm outputs the same interpolant. For the variable condition note that $V(\neg A) = V(A)$. Now, we have to consider the last rule used in the proof. We investigate some cases and the rest are similar and left to the reader. Suppose the last rule in $\pi$ is $(L \vee)$. Then, based on the position of the main formula in the Maehara partition in the conclusion, we have two cases. Let us only consider the case where the main formula is on the left-hand side of the Maehara partition. Then, $\pi$ has the following form:
\begin{prooftree}
\AxiomC{$\pi_1$}
\noLine
\UnaryInfC{$B, \Gamma_1 ; \Gamma_2 \overset{I}{\Longrightarrow} \Delta_1; \Delta_2, \neg A$}
\AxiomC{$\pi_2$}
\noLine
\UnaryInfC{$C, \Gamma_1 ; \Gamma_2 \overset{J}{\Longrightarrow} \Delta_1; \Delta_2, \neg A$}
\RightLabel{\small $(R \vee)$}
\BinaryInfC{$B \vee C, \Gamma_1 ; \Gamma_2 \overset{I \vee J}{\Longrightarrow} \Delta_1; \Delta_2, \neg A$}
\end{prooftree}
By the induction hypothesis, there are $\pi_1'$ and $\pi_2'$ satisfying the conditions of the lemma. Then, we take the following proof as $\pi'$:
\begin{prooftree}
\AxiomC{$\pi_1'$}
\noLine
\UnaryInfC{$B, \Gamma_1 ; \Gamma_2, A \overset{I}{\Longrightarrow} \Delta_1; \Delta_2$}
\AxiomC{$\pi_2'$}
\noLine
\UnaryInfC{$C, \Gamma_1 ; \Gamma_2, A \overset{J}{\Longrightarrow} \Delta_1; \Delta_2$}
\RightLabel{\small $(R \vee)$}
\BinaryInfC{$B \vee C, \Gamma_1 ; \Gamma_2 , A \overset{I \vee J}{\Longrightarrow} \Delta_1; \Delta_2$}
\end{prooftree}
It is easy to that the interpolant does not change and satisfies the variable condition. The case where the main formula $B \vee C$ is on the right-hand side of the Maehara partition is similar.

Suppose the last rule is a monochromatic cut rule. Again, based on the position of the cut formula there are two cases. Let the cut formula be on the right-hand side of the Maehara partition:
\begin{prooftree}
\AxiomC{$\pi_1$}
\noLine
\UnaryInfC{$\Gamma_1 ; \Gamma_2 \overset{I}{\Longrightarrow} \Delta_1; \Delta_2, B, \neg A$}
\AxiomC{$\pi_2$}
\noLine
\UnaryInfC{$\Gamma_1 ; \Gamma_2, B \overset{J}{\Longrightarrow} \Delta_1; \Delta_2, \neg A$}
\RightLabel{\small $(cut)$}
\BinaryInfC{$ \Gamma_1 ; \Gamma_2 \overset{I \wedge J}{\Longrightarrow} \Delta_1; \Delta_2, \neg A$}
\end{prooftree}
By the induction hypothesis, there are proofs $\pi_1'$ and $\pi_2'$ satisfying the required conditions. Then, we take the following proof as $\pi'$:
\begin{prooftree}
\AxiomC{$\pi_1'$}
\noLine
\UnaryInfC{$\Gamma_1 ; \Gamma_2, A \overset{I}{\Longrightarrow} \Delta_1; \Delta_2, B$}
\AxiomC{$\pi_2'$}
\noLine
\UnaryInfC{$\Gamma_1 ; \Gamma_2, B, A \overset{J}{\Longrightarrow} \Delta_1; \Delta_2$}
\RightLabel{\small $(cut)$}
\BinaryInfC{$\Gamma_1 ; \Gamma_2 , A \overset{I \wedge J}{\Longrightarrow} \Delta_1; \Delta_2$}
\end{prooftree}
The case where the cut formula is on the left-hand side of the Maehara partition is similar.

Assume the last rule used in $\pi$ is $(R \neg)$. Then, there are three cases: either $\neg A$ is the main formula or the main formula is not $\neg A$ and it is either on the left-hand side or right-hand side of the Maehara partition. First, let $\neg A$ be the main formula. Then, $\pi$ has the following form:
\begin{prooftree}
\AxiomC{$\pi_1$}
\noLine
\UnaryInfC{$\Gamma_1 ; \Gamma_2,  A \overset{I}{\Longrightarrow} \Delta_1; \Delta_2$}
\RightLabel{\small $(R \neg)$}
\UnaryInfC{$\Gamma_1 ; \Gamma_2  \overset{I}{\Longrightarrow} \Delta_1; \Delta_2, \neg A$}
\end{prooftree}
Then, we take $\pi'=\pi_1$. 

Now, let the main formula be different from $\neg A$ and be on the left-hand side of the Maehara partition. Then $\pi$ has the form:
\begin{prooftree}
\AxiomC{$\pi_1$}
\noLine
\UnaryInfC{$\Gamma_1, B ; \Gamma_2 \overset{I}{\Longrightarrow} \Delta_1; \Delta_2, \neg A$}
\RightLabel{\small $(R \neg)$}
\UnaryInfC{$\Gamma_1 ; \Gamma_2  \overset{I}{\Longrightarrow} \Delta_1, \neg B; \Delta_2, \neg A$}
\end{prooftree}
By the induction hypothesis, there exists a proof  $\pi_1'$ satisfying the required condition.
Then, we take the following as the proof $\pi'$:
\begin{prooftree}
\AxiomC{$\pi_1'$}
\noLine
\UnaryInfC{$\Gamma_1, B ; \Gamma_2, A \overset{I}{\Longrightarrow} \Delta_1; \Delta_2$}
\RightLabel{\small $(R \neg)$}
\UnaryInfC{$\Gamma_1 ; \Gamma_2 , A \overset{I}{\Longrightarrow} \Delta_1, \neg B; \Delta_2$}
\end{prooftree}
The case where the main formula is on the right-hand side of the Maehara partition is similar.

Finally, it is clear from the construction of $\pi'$ that its length is linear in the length of $\pi$.
\end{proof}
Note that in the proof of Lemma \ref{lem.neginversion} to provide $\pi'$ for item $(1)$, we used induction on the length of $\pi$ in item $(1)$ and there was no need to use other items, namely $(2)$, $(3)$, or $(4)$. In other words, each case is handled by induction on itself.

\LemPruneClauseSet*
\begin{proof}
It suffices to treat the case $L_\mathrm{D} = \{ p \}$. The result then follows by induction.
For \ref{prune1}, we proceed as follows. First, delete all the clauses in 
$\calA$ containing $\top$. Then, as the next step, take a clause $C \in \calA$. If both $p$ and $\neg p$ have occurred in $C$, then delete $C$ from $\calA$. In particular, if $\calA$ has only one clause containing both $p$ and $\neg p$ then the result is $\{\}$. Repeat this step for every clause in $\calA$ until there are no more clauses in which both $p$ and $\neg p$ occur. Now, we resolve against $p$, i.e., we define the clause set $\calA^*$ as the following set of clauses:

\vspace{5pt}
$ \begin{cases}
C \in \calA & p \notin C \; \text{and} \; \neg p \notin C \\ 
D_1 \cup D_2 \setminus \{p, \neg p\} & D_1 \cup \{p\} \in \calA \; \text{and} \; D_2 \cup \{\neg p\} \in \calA
\end{cases} $
\vspace{5pt}

\noindent Note that as a result of the previous step, $\neg p \notin D_1$ and $p \notin D_2$. In the end, neither $p$ nor $\neg p$ appears in $\calA^*$ and it is a clause set in the language $L \setminus L_{\mathrm{D}}$. Now, we want to prove $I \models \calA^*$ which means that we have to prove  $I \models C$ for any $C \in \calA^*$. There are two cases: if $C \in \calA$, then as $I \models \calA$, it means that  $I \models C$. The second case is when $C$ has the form $D_1 \cup D_2 \setminus \{p, \neg p\}$ where $D_1 \cup \{p\} \in \calA$ and $D_2 \cup \{\neg p\} \in \calA$. Therefore, $I \models D_1 \cup \{p\}$ and $I \models D_2 \cup \{\neg p\}$. Then, we claim that $I \models D_1 \cup D_2$. Because, if $I \models p$ then $I \models D_2$, and if $I \models \neg p$ then $I \models D_1$. Hence $I \models \calA^*$.

For \ref{prune2}, again it suffices to treat the case $L_{\mathrm{D}}=\{p\}$. 
Suppose $I' \models \calA^*$. On the language $L \setminus L_{\mathrm{D}}$, define $I$ as $I'$. For the clause $C \in \calA$ which contains $\top$ or contains both $p$ and $\neg p$, clearly we have $I \models C$. Now, to determine the valuation of $p$ in $I$ take a clause $D \cup \{\ell\} \in \calA$, where $\ell$ is either $p$ or $\neg p$. If $I' \not \models D$ then define $I(\ell)=\top$. Otherwise, if $I' \models D$ then move to another clause containing $p$ or $\neg p$ in $\calA$. Repeat this process until for one $D \cup \{\ell\} \in \calA$, $I' \not \models D$ and then $I(\ell)=\top$. Otherwise, if for all $D \cup \{\ell\} \in \calA$, $I' \models D$, then w.l.o.g. define $I(p)=\top$. Clearly, $I$ is an extension of $I'$ and by its construction $I \models \calA$.
\end{proof}

\ObsSubs*
\begin{proof}
For~(\ref{obs.subs.subset}) it suffices to observe that for all $B\in \calB \subseteq \calA$ we can take $A = B \in \calA$.
For~(\ref{obs.subs.trans}) let $C \in \calC$.
Then there is a $B\in\calB$ with $B\subseteq C$.
Hence, there is an $A\in \calA$ with $A\subseteq B \subseteq C$.
For~(\ref{obs.subs.compat_union}) let $B\in \calB$.
Then there is an $A\in \calA \subseteq \calA\cup\calC$
 with $A\subseteq B$.
Moreover, for $C\in\calC$ we can simply take $C\in \calC \subseteq \calA \cup\calC$.
For~(\ref{obs.subs.compat_times}) let $B\cup C \in \calB \times \calC$, then
 there is an $A\in \calA$ with $A\subseteq B$.
So $A\cup C \in \calA\times \calC$ with $A\cup C \subseteq B\cup C$.
For~(\ref{obs.subs.distand}) we make a case distinction on the form
 of a $D \in (\calA\cup\calC)\times (\calB\cup\calC)$.
If $D=A\cup B$ for $A\in \calA$, $B\in\calB$ then $D\in \calA\times\calB \subseteq (\calA\times\calB)\cup \calC$.
Otherwise there is a $C\in \calC$ with $C\subseteq D$ and we are done.
\end{proof}

\LemWReduce*
\begin{proof}
We obtain $\pi'$ from $\pi$ by shifting weakening inferences upwards until
 they are in a permitted position.
If a weakening inference is below a unary inference as in $\chi =$
\[
\infer[(Rw)]{\Mseq{\Gamma_1}{\Gamma_2, \neg B}{\Delta_1}{\Delta_2, A}}{
  \infer[(L\neg)]{\Mseq{\Gamma_1}{\Gamma_2, \neg B}{\Delta_1}{\Delta_2}}{ 
    \deduce{\Mseq{\Gamma_1}{\Gamma_2}{\Delta_1}{\Delta_2,B}}{(\chi_1)}
  }
}
\]
we replace $\chi$ by $\chi' =$
\[
\infer[(L\neg)]{\Mseq{\Gamma_1}{\Gamma_2, \neg B}{\Delta_1}{\Delta_2, A}}{
  \infer[(Rw)]{\Mseq{\Gamma_1}{\Gamma_2}{\Delta_1}{\Delta_2, B, A}}{
    \deduce{\Mseq{\Gamma_1}{\Gamma_2}{\Delta_1}{\Delta_2,B}}{(\chi_1)}
  }
}
\]
The other unary rules are treated analogously.

If a weakening inference is below a binary inference as in $\chi =$
\[
\infer[(Rw)]{\Mseq{\Gamma_1}{\Gamma_2}{\Delta_1}{\Delta_2, B\land C, A}}{
  \infer[(R\land)]{\Mseq{\Gamma_1}{\Gamma_2}{\Delta_1}{\Delta_2, B\land C}}{
    \deduce{\Mseq{\Gamma_1}{\Gamma_2}{\Delta_1}{\Delta_2, B}}{(\chi_1)}
    &
    \deduce{\Mseq{\Gamma_1}{\Gamma_2}{\Delta_1}{\Delta_2, C}}{(\chi_2)}
  }
}
\]
we replace $\chi$ by $\chi' =$
\[
\infer[(R\land)]{\Mseq{\Gamma_1}{\Gamma_2}{\Delta_1}{\Delta_2, B\land C, A}}{
  \infer[(Rw)]{\Mseq{\Gamma_1}{\Gamma_2}{\Delta_1}{\Delta_2, B, A}}{
    \deduce{\Mseq{\Gamma_1}{\Gamma_2}{\Delta_1}{\Delta_2, B}}{(\chi_1)}
  }
  &
  \infer[(Rw)]{\Mseq{\Gamma_1}{\Gamma_2}{\Delta_1}{\Delta_2, C, A}}{
    \deduce{\Mseq{\Gamma_1}{\Gamma_2}{\Delta_1}{\Delta_2, C}}{(\chi_2)}
  }
}
\]
The other binary rules are treated analogously.

We have $\calM(\chi') = \calM(\chi)$ which, by induction and
 Observation~\ref{obs.subs}/(\ref{obs.subs.compat_union}) and (\ref{obs.subs.compat_times}), entails~(\ref{lem.wreduce.ipol}): $\calM(\pi') = \calM(\pi)$.
(\ref{lem.wreduce.Rcut}) is evident from the definition of the transformation.
Let $\mu$ be the indicated occurrence of $A$ in the end-sequent of $\chi$
 and let $\mu'$ be the corresponding occurrence in the end-sequent of
 $\chi'$.
Then since both $\mu$ and $\mu'$ are weak, we have $\weight(\mu') = \weight(\mu) = 0$.
The weight of the other formula occurrences in the end-sequent does
 not change.
This shows~(\ref{lem.wreduce.focc.weight}).
Moreover, note that every axiom of $\chi'$ in which an
 ancestor of a formula occurrence $\nu'$ in the end-sequent of $\chi'$
 is active is also an axiom of $\chi$ in which an ancestor of the corresponding
 formula occurrence $\nu$ in the end-sequent of $\chi$ is active.
This shows~(\ref{lem.wreduce.focc.ax}) and, together with the observation that
 tameness condition~(\ref{def.tame.Omega}) is preserved, it shows~(\ref{lem.wreduce.tame}).
\end{proof}

\LemCEIpol*
\begin{proof}
Here we treat the cases missing in the main text.

{\bf Permutation of a unary inference over the cut:}
If $\chi$ is of the form
\[
\infer[\mathrm{cut}_c]{\Mseq{\Gamma_1}{\Gamma_2}{\Delta_1, \neg A}{\Delta_2}}{
  \infer{\Mseq{\Gamma_1}{\Gamma_2}{\Delta_1, \neg A}{\Delta_2, C}}{
    \deduce{\Mseq{\Gamma_1, A}{\Gamma_2}{\Delta_1}{\Delta_2, C_{\mu_1}}}{(\chi_1)}
  }
  &
  \deduce{\Mseq{\Gamma_1}{\Gamma_2, C_{\mu_2}}{\Delta_1, \neg A}{\Delta_2}}{(\chi_2)}
}
\]
we define $\chi'_1 =$
\[
\infer{\Mseq{\Gamma_1, A}{\Gamma_2}{\Delta_1, \neg A}{\Delta_2, C}}{
  \deduce{\Mseq{\Gamma_1, A}{\Gamma_2}{\Delta_1}{\Delta_2, C}}{(\chi_1)}
}
\]
and apply Lemma~\ref{lem.wreduce} to $\chi'_1$ to obtain an {\LKminus} proof
 $\chi^*_1$ with~(\ref{lem.wreduce.ipol}) and~(\ref{lem.wreduce.focc}).
Moreover, we define $\chi'_2 =$
\[
\infer{\Mseq{\Gamma_1, A}{\Gamma_2, C}{\Delta_1, \neg A}{\Delta_2}}{
  \deduce{\Mseq{\Gamma_1}{\Gamma_2, C}{\Delta_1, \neg A}{\Delta_2}}{(\chi_2)}
}
\]
and apply Lemma~\ref{lem.wreduce} to $\chi'_2$ to obtain an {\LKminus} proof
 $\chi^*_2$ with~(\ref{lem.wreduce.ipol}) and~(\ref{lem.wreduce.focc}).
Finally, we define $\chi^* =$
\[
\infer[\mathrm{cut}_{c'}]{\Mseq{\Gamma_1, A}{\Gamma_2}{\Delta_1, \neg A}{\Delta_2}}{
  \deduce{\Mseq{\Gamma_1, A}{\Gamma_2}{\Delta_1, \neg A}{\Delta_2, C}}{(\chi_1^*)}
  &
  \deduce{\Mseq{\Gamma_1, A}{\Gamma_2, C}{\Delta_1, \neg A}{\Delta_2}}{(\chi_2^*)}
}
\]
For~\ref{lem.ce_ipol.tamewred} we observe that $\pi[\chi^*]$ is w-reduced
 and $\chi, \chi^*$ satisfy (*), so $\pi[\chi^*]$ is tame.
\ref{lem.ce_ipol.weight} follows from $\weight(c) = \weight(\mu_1) + \weight(\mu_2) + 1$ and $\weight(c') = \weight(\mu_1) + \weight(\mu_2)$.
For~\ref{lem.ce_ipol.M} note that we have $\CNF(\calM(\chi))
 = \CNF(\calM(\chi_1)\land\calM(\chi_2)) = \CNF(\calM(\chi_1^*) \land \calM(\chi_2^*)) = \CNF(\calM(\chi^*))$.
We proceed analogously for the other unary inferences.

{\bf Reduction of contraction:} If $\chi =$
\[
\infer[\mathrm{cut}_c]{\Mseq{\Gamma_1}{\Gamma_2}{\Delta_1}{\Delta_2}}{
  \infer{\Mseq{\Gamma_1}{\Gamma_2}{\Delta_1}{\Delta_2, C}}{
    \deduce{\Mseq{\Gamma_1}{\Gamma_2}{\Delta_1}{\Delta_2, C_{\mu_{1,1}}, C_{\mu_{1,2}}}}{(\chi_1)}
  }
  &
  \deduce{\Mseq{\Gamma_1}{\Gamma_2,C_{\mu_2}}{\Delta_1}{\Delta_2}}{(\chi_2)}
}
\]
we first treat the special case
 that $\chi_1$ ends with a weakening inference as in:
\[
\infer{\Mseq{\Gamma_1}{\Gamma_2}{\Delta_1}{\Delta_2}}{
  \infer{\Mseq{\Gamma_1}{\Gamma_2}{\Delta_1}{\Delta_2, C}}{
    \infer[\mathrm{wax}]{\Mseq{\Gamma_1}{\Gamma_2}{\Delta_1}{\Delta_2,C,C}}{}
  }
  &
  \deduce{\Mseq{\Gamma_1}{\Gamma_2,C}{\Delta_1}{\Delta_2}}{(\chi_2)}
}
\]
Then at least one $\mu_{1,1}$ and $\mu_{1,2}$ is weak because $\chi_1$ starts
 with an axiom.
Moreover, at most one of them is weak since otherwise the main formula of
 the contraction would be weak.
So, more generally, if one of $\mu_{1,1},\mu_{1,2}$, say w.l.o.g.\ $\mu_{1,1}$,
 is weak we define $\chi^*_1$ from $\chi_1$ by deleting $\mu_{1,1}$, all its
 ancestors, and all weakeing inferences that introduce these ancestors.
We replace $\chi$ by $\chi^* =$
\[
\infer[\mathrm{cut}_{c'}]{\Mseq{\Gamma_1}{\Gamma_2}{\Delta_1}{\Delta_2}}{
  \deduce{\Mseq{\Gamma_1}{\Gamma_2}{\Delta_1}{\Delta_2, C}}{(\chi^*_1)}
  &
  \deduce{\Mseq{\Gamma_1}{\Gamma_2, C}{\Delta_1}{\Delta_2}}{(\chi^*_2)}
}
\]
Then we have: $\weight(c) > \weight(c')$, $\pi[\chi^*]$ is tame and
 w-reduced by (*), and $\calM(\chi^*) = \calM(\chi)$.

So we can assume that $\chi_1$ does not end with a weakening inference
 and that none of $\mu_{1,1}$ and $\mu_{1,2}$ is weak.
Then we define $\chi'_2 = $
\[
\infer{\Mseq{\Gamma_1}{\Gamma_2, C}{\Delta_1}{\Delta_2,C}}{
  \deduce{\Mseq{\Gamma_1}{\Gamma_2, C}{\Delta_1}{\Delta_2}}{(\chi_2)}
}
\]
and we apply Lemma~\ref{lem.wreduce} to $\chi'_2$ to obtain an {\LKminus} proof
 $\chi^*_2$ with~(\ref{lem.wreduce.ipol}) and~(\ref{lem.wreduce.focc}).
We define $\chi^* =$
\[
\infer[c'']{\Mseq{\Gamma_1}{\Gamma_2}{\Delta_1}{\Delta_2}}{
  \infer[c']{\Mseq{\Gamma_1}{\Gamma_2}{\Delta_1}{\Delta_2,C}}{
    \deduce{\Mseq{\Gamma_1}{\Gamma_2}{\Delta_1}{\Delta_2,C,C}}{(\chi_1)}
    &
    \deduce{\Mseq{\Gamma_1}{\Gamma_2, C}{\Delta_1}{\Delta_2, C}}{(\chi^*_2)}
  }
  & \hspace*{-20pt}
  \deduce{\Mseq{\Gamma_1}{\Gamma_2,C}{\Delta_1}{\Delta_2}}{(\chi_2)}
}
\]
For~\ref{lem.ce_ipol.tamewred} we observe that $\pi[\chi^*]$ is w-reduced.
Moreover, $\chi,\chi^*$ satisfy (*), so $\pi[\chi^*]$ is tame.
For~\ref{lem.ce_ipol.weight} we have $\weight(c) = 1 + \weight(\mu_{1,1}) + \weight(\mu_{1,2}) + \weight(\mu_2)$,
$\weight(c') = \weight(\mu_{1,1}) + \weight(\mu_2)$,
and $\weight(c'') = 1+ \weight(\mu_{1,2}) + \weight(\mu_2)$.
Then $\weight(c') < \weight(c)$.
Moreover, since $\chi_1$ does not end with a weakening inference
 and $\mu_{1,1}$ is not weak, we have $\weight(\mu_{1,1}) > 0$ and therefore
 $\weight(c'') < \weight(c)$.
For~\ref{lem.ce_ipol.M} we have  $\CNF(\calM(\chi))
= \CNF(\calM(\chi_1) \land \calM(\chi_2))
=^\text{Obs.~\ref{obs.CNF}/(\ref{obs.CNF.comm}),(\ref{obs.CNF.assoc}),(\ref{obs.CNF.andidem})}
 \CNF((\calM(\chi_1) \land \calM(\chi_2)) \land \calM(\chi_2))
 = \CNF(\calM(\chi^*))$.

{\bf Reduction of the degree:} If $\chi$ is of the form
\[
\infer{\Mseq{\Gamma_1}{\Gamma_2}{\Delta_1}{\Delta_2}}{
  \infer{\Mseq{\Gamma_1}{\Gamma_2}{\Delta_1}{\Delta_2, A\land B}}{
    \deduce{\Mseq{\Gamma_1}{\Gamma_2}{\Delta_1}{\Delta_2, A}}{(\chi_1)}
    &
    \deduce{\Mseq{\Gamma_1}{\Gamma_2}{\Delta_1}{\Delta_2,B}}{(\chi_2)}
  }
  &
  \infer{\Mseq{\Gamma_1}{\Gamma_2, A\land B}{\Delta_1}{\Delta_2}}{
    \deduce{\Mseq{\Gamma_1}{\Gamma_2, A}{\Delta_1}{\Delta_2}}{(\chi_3)}
  }
}
\]
we replace it by $\chi^* =$
\[
\infer{\Mseq{\Gamma_1}{\Gamma_2}{\Delta_1}{\Delta_2}}{
  \deduce{\Mseq{\Gamma_1}{\Gamma_2}{\Delta_1}{\Delta_2, A}}{(\chi_1)}
  &
  \deduce{\Mseq{\Gamma_1}{\Gamma_2,A}{\Delta_1}{\Delta_2}}{(\chi_3)}
}
\]
For \ref{lem.ce_ipol.tamewred} we observe that $\pi[\chi^*]$ is w-reduced
 and, by (*), also tame.
For \ref{lem.ce_ipol.weight} we observe that $\degree(c) > \degree(c')$.
For \ref{lem.ce_ipol.M} we have
$\CNF(\chi) 
= \CNF((\calM(\chi_1) \land \calM(\chi_2)) \land \calM(\chi_3))
\subs^\text{Obs.~\ref{obs.subs}/(\ref{obs.subs.subset})}  
 \CNF(\calM(\chi_1) \land \calM(\chi_3)) = \CNF(\chi^*)$.
The cases of the other connectives are analogous.

\end{proof}

\end{document}